\documentclass{sig-alternate-10pt}

\usepackage{cite}
\usepackage{wrapfig}
\usepackage{graphicx}
\usepackage{subfigure}
\usepackage{grffile}
\usepackage{xcolor}
\usepackage{fancybox}
\usepackage{calc}
\usepackage{url}
\usepackage{enumitem}
\usepackage{mdwlist}
\usepackage{authblk}
\makeatletter
\renewcommand\AB@affilsepx{ \hspace{3em} \protect\Affilfont}
\makeatother


\newcommand{\cut}[1]{}
\newcommand{\fixme}[1]{}

\newcommand{\new}[1]{#1}

\definecolor{termcolor}{RGB}{102,138,201}

\newtheorem{theorem}{Theorem}
\newtheorem{conjecture}{Conjecture}
\newdef{definition}{Definition}

\begin{document}

\title{Low Latency via Redundancy}

\author[1]{Ashish Vulimiri}
\author[1]{P. Brighten Godfrey}
\author[2]{Radhika Mittal}
\author[2]{Justine Sherry}
\author[2]{Sylvia Ratnasamy}

\author[2,3]{\\Scott Shenker}
\affil[1]{UIUC}
\affil[2]{UC Berkeley}
\affil[3]{ICSI}


\maketitle

\begin{abstract}
\textnormal{Low latency is critical for interactive networked applications.  But while we know how to scale systems to increase capacity, reducing latency --- especially the tail of the latency distribution --- can be much more difficult.  In this paper, we argue that the use of redundancy is an effective way to convert extra capacity into reduced latency.  By initiating redundant operations across diverse resources and using the first result which completes, redundancy improves a system's latency even under exceptional conditions.  We study the tradeoff with added system utilization, characterizing the situations in which replicating all tasks reduces mean latency.  We then demonstrate empirically that replicating all operations can result in significant mean and tail latency reduction in real-world systems including DNS queries, database servers, and packet forwarding within networks.}
\end{abstract}

\section{Introduction}
\label{sec:intro}


Low latency is important for humans.  Even slightly higher web page load times can significantly reduce visits from users and revenue, as demonstrated by several sites~\cite{souders09velocity}.  For example, injecting just $400$ milliseconds of artificial delay into Google search results caused the delayed users to perform $0.74$\% fewer searches after 4-6 weeks~\cite{brutlag09}.  A $500$ millisecond delay in the Bing search engine reduced revenue per user by $1.2$\%, or $4.3$\% with a $2$-second delay~\cite{souders09velocity}.  Human-computer interaction studies similarly show that people react to small differences in the delay of operations (see \cite{gray2000milliseconds} and references therein).

Achieving consistent low latency is challenging.  Modern applications are highly distributed, and likely to get more so as cloud computing separates users from their data and computation.  Moreover, application-level operations often require tens or hundreds of tasks to complete --- due to many objects comprising a single web page~\cite{ramachandran10web}, or aggregation of many back-end queries to produce a front-end result~\cite{dixon09shopzilla,dctcp}.  This means individual tasks may have latency budgets on the order of a few milliseconds or tens of milliseconds, and \emph{the tail} of the latency distribution is critical.  Such outliers are difficult to eliminate because they have many sources in complex systems; even in a well-provisioned system where individual operations usually work, some amount of uncertainty is pervasive. Thus, latency is a difficult challenge for networked systems:  How do we make the other side of the world feel like it is \emph{right here}, even under exceptional conditions?

One powerful technique to reduce latency is \emph{redundancy}: Initiate an operation multiple times, using as diverse resources as possible, and use the first result which completes.  Consider a host that queries multiple DNS servers in parallel to resolve a name.  The overall latency is the minimum of the delays across each query, thus potentially reducing both the mean and the tail of the latency distribution.  For example, a replicated DNS query could mask spikes in latency due to a cache miss, network congestion, packet loss, a slow server, and so on.  The power of this technique is that it reduces latency precisely under the most challenging conditions---when delays or failures are unpredictable---and it does so without needing any information about what these conditions might be.

Redundancy has been employed to reduce latency in several networked systems: notably, as a way to deal with failures in DTNs~\cite{jain2005using}, in a multi-homed web proxy overlay~\cite{andersen05monet}, and in limited cases in distributed job execution frameworks \cite{Foster72, ananthanarayanan12dolly,Zaharia2008}.

However, these systems are exceptions rather than the rule. Redundant queries are typically eschewed, whether across the Internet or within data centers.  The reason is rather obvious: duplicating every operation doubles system utilization, or increases usage fees for bandwidth and computation.  The default assumption in system design is that doing less work is best.

But when exactly is that natural assumption valid? Despite the fact that redundancy is a fundamental technique that has been used in certain systems to reduce latency, the conditions under which it is effective are not well understood --- and we believe as a result, it is not widely used.

\medskip

In this paper, we argue that redundancy is an effective \emph{general} technique to achieve low latency in networked systems. Our results show that redundancy could be used much more commonly than it is, and in many current systems represents a missed opportunity. 

Making that argument requires an understanding of when replication improves latency and when it does not. Consider a system with a fixed set of servers, in which queries are relatively inexpensive for clients to send.  If a single client duplicates its queries, its latency is likely to decrease, but it also affects other users in the system to some degree.  If \emph{all} clients duplicate every query, then every client has the benefit of receiving the faster of two responses (thus decreasing mean latency) but system utilization has doubled (thus increasing mean latency).  It is not immediately obvious under what conditions the former or latter effect dominates.

\medskip

Our first key contribution is to characterize when such global redundancy improves latency.  We introduce a queueing model of query replication, giving an analysis of the expected response time as a function of system utilization and server-side service time distribution.  Our analysis and extensive simulations demonstrate that assuming the client-side cost of replication is low, there is a server-side \emph{threshold load} below which replication always improves latency.  We give a crisp conjecture, with substantial evidence, that this threshold \emph{always lies between 25\% and 50\% utilization regardless of the service time distribution}, and that it can approach $50\%$ arbitrarily closely as variance in service time increases. \new{Our results indicate that redundancy should have a net positive impact in a large class of systems, despite the extra load that it adds.}
%

\medskip

Our second key contribution is to demonstrate multiple practical application scenarios in which replication empirically provides substantial benefit, yet is not generally used today.  These scenarios, along with scenarios in which replication is \emph{not} effective, corroborate the results of our analysis.  More specifically:

\begin{itemize}
\item \textbf{DNS queries across the wide area.} Querying multiple DNS servers reduces the fraction of responses later than $500$ ms by $6.5\times$, while the fraction later than $1.5$ sec is reduced by $50\times$, compared with a non-replicated query to the \emph{best} individual DNS server. This improvement is more than an order of magnitude better than the estimated \emph{threshold cost-effectiveness} --- that is, replication saves more than $100$ msec per KB of added traffic --- indicating that replication is useful in practice.  Similarly, a simple analysis indicates that replicating TCP connection establishment packets can save roughly $170$ msec (in the mean) and $880$ msec (in the tail) per KB of added traffic.
\item \textbf{Database queries within a data center.} We implement query replication in a database system similar to a web service, where a set of clients continually read objects from a set of back-end servers.  Our results indicate that when most queries are served from disk and file sizes are small, replication provides substantial latency reduction of up to $2\times$ in the mean and up to $8\times$ in the tail. As predicted by our analysis, mean latency is reduced up to a server-side \emph{threshold load} of $30$-$40\%$.
We also show that when retrieved files become large or the database resides in memory, replication does not offer a benefit.  This occurs across both a web service database and the memcached in-memory database, and is consistent with our analysis: in both cases (large or in-memory files), the client-side cost of replication becomes significant \emph{relative to} the mean query latency.
\item \textbf{In-network packet replication.} We design a simple strategy for switches, to replicate the initial packets of a flow but treat them as lower priority.  This offers an alternate mechanism to limit the negative effect of increased utilization, and simulations indicate it can yield up to a $38\%$ median end-to-end latency reduction for short flows.
\end{itemize}

In summary, as system designers we typically build scalable systems by avoiding unnecessary work.  The significance of our results is to characterize a large class of cases in which duplicated work is a useful and elegant way to achieve robustness to variable conditions and thus reduce latency.


\section{System view}
\label{sec:system-view}

In this section we characterize the tradeoff between the benefit (fastest of multiple options) and cost (doing more work) due to redundancy from the perspective of a system designer optimizing a \emph{fixed set} of resources.  We analyze this tradeoff in an abstract queueing model (\S\ref{sec:queueing-model}) and evaluate it empirically in two applications: a disk-backed database (\S\ref{sec:apache}) and an in-memory cache (\S\ref{sec:memcached}).

\S\ref{sec:individual-view} considers the scenario where the available resources are provisioned according to payment, rather than static.



\subsection{System view: Queueing analysis}
\label{sec:queueing-model}
\label{APPAREF}

\begin{figure*}[ht]
\centering
  \subfigure[Mean: deterministic]{
    \includegraphics[width=0.3\textwidth]{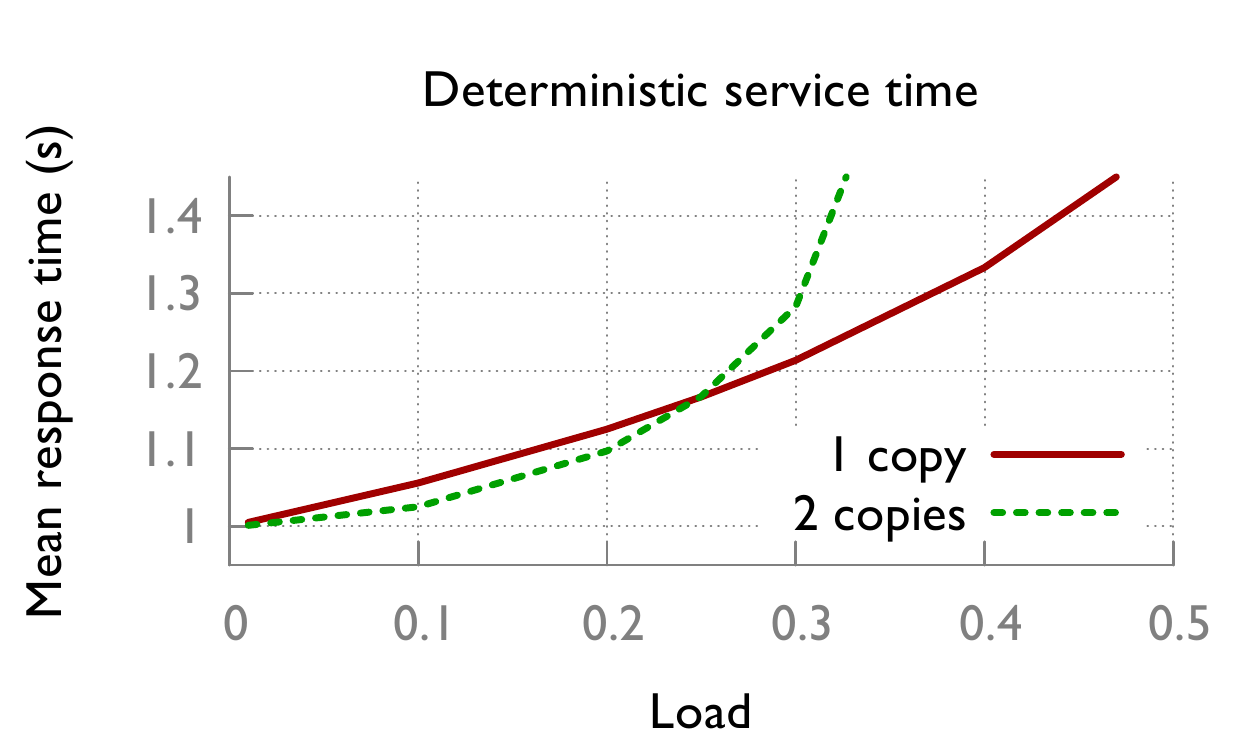}
    \label{fig:md1}
  }
  \subfigure[Mean: Pareto]{
    \includegraphics[width=0.3\textwidth]{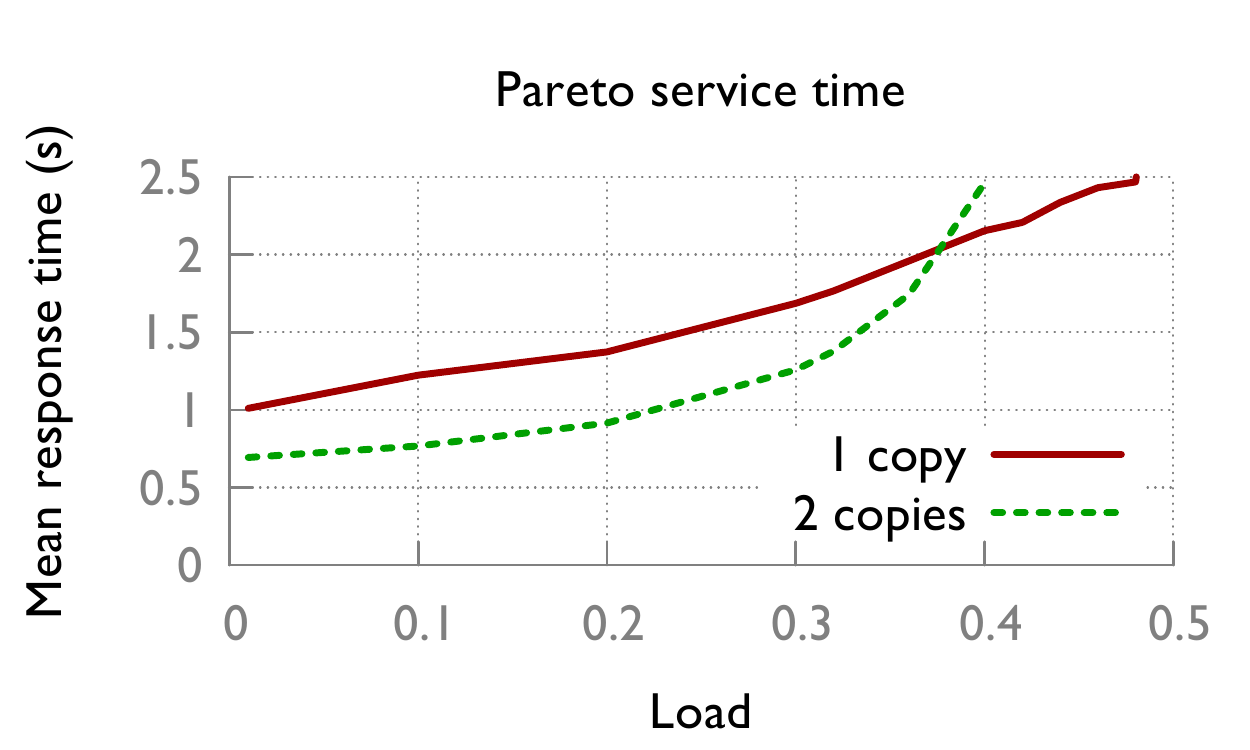}
    \label{fig:mpar1:mean}
  }
  \subfigure[CDF: Pareto]{
    \includegraphics[width=0.3\textwidth]{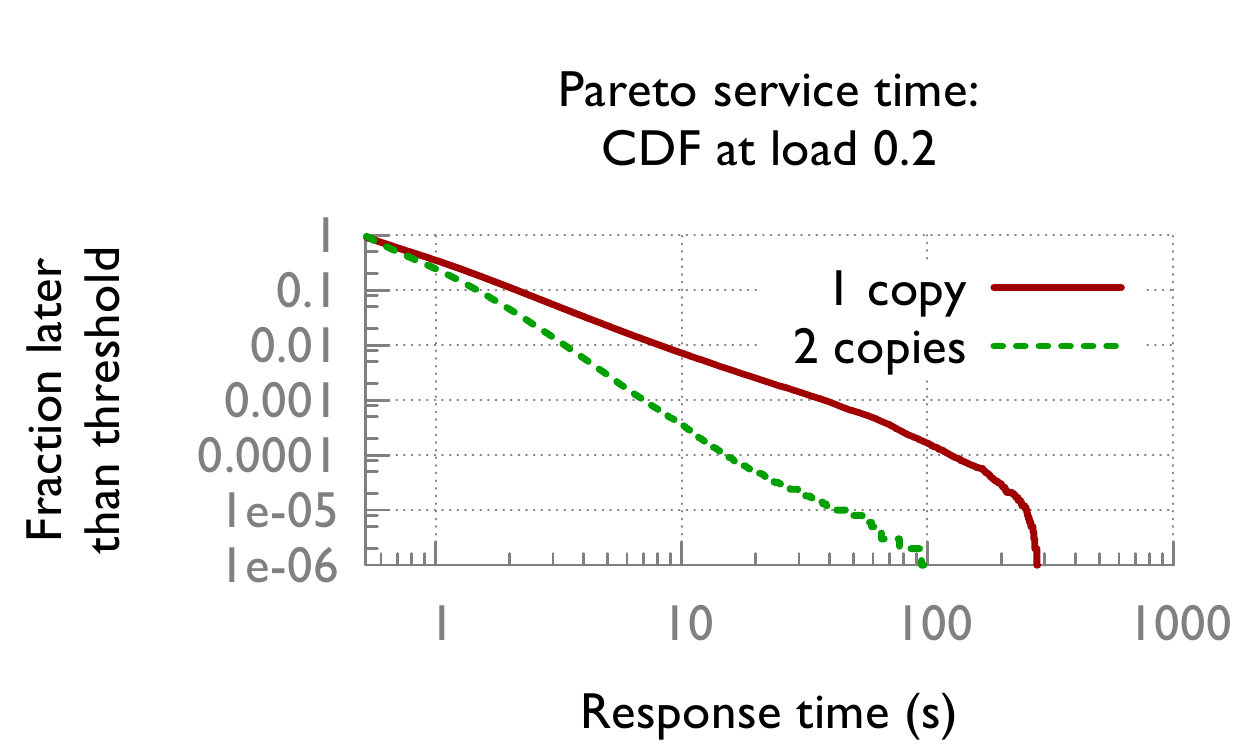}
    \label{fig:mpar1:cdf}
  }
  \caption{A first example of the effect of replication, showing response times when service time distribution is deterministic and Pareto ($\alpha = 2.1$)}
  \label{fig:queue-examples}
\end{figure*}

Two factors are at play in a system with redundancy. Replication reduces latency by taking the faster of two (or more) options to complete, but it also worsens latency by increasing the overall utilization. In this section, we study the interaction between these two factors in an abstract queueing model.

We assume a set of $N$ independent, identical servers, each with the same service time distribution $S$. Requests arrive in the system according to a Poisson process, and $k$ copies are made of each arriving request and enqueued at $k$ of the $N$ servers, chosen uniformly at random. To start with, we will assume that redundancy is ``free'' for the clients --- that it adds no appreciable penalty apart from an increase in server utilization. We consider the effect of client-side overhead later in this section.

Figures~\ref{fig:md1} and \ref{fig:mpar1:mean} show results from a simulation of this queueing model, measuring the mean response time (queueing delay $+$ service time) as a function of load with two different service time distributions.  Replication improves the mean, but provides the greatest benefit in the tail, for example reducing the $99.9$th percentile by $5\times$ under Pareto service times. Note the thresholding effect: in both systems, there is a \emph{threshold load} below which redundancy always helps improve latency, but beyond which the extra load it adds overwhelms any latency reduction that it achieves. The threshold is higher --- i.e., redundancy helps over a larger range of loads --- when the service time distribution is more variable.

The threshold load, defined formally as the largest utilization below which replication always helps mean response time, will be our metric of interest in this section. \new{We investigate the effect of the service time distribution on the threshold load both analytically and in simulations of the queueing model. Our results, in brief:}

\begin{enumerate}
\item If redundancy adds no client-side cost (meaning server-side effects are all that matter), there is strong evidence to suggest that no matter what the service time distribution, the threshold load has to be more than $25\%$.
\item In general, the higher the variability in the service-time distribution, the larger the performance improvement achieved.
\item Client-side overhead can diminish the performance improvement due to redundancy. In particular, the threshold load can go below $25\%$ if redundancy adds a client-side processing overhead that is significant compared to the server-side service time.
\end{enumerate}

\subsubsection*{If redundancy adds no client-side cost}

We start with a simple, analytically-tractable special case: when the service times at each server are exponentially distributed. A closed form expression for the response time CDF exists in this case, and it can be used to establish the following result.

\begin{theorem}
\label{thm:exponential}
If the service times at every server are i.i.d.\ exponentially distributed, the threshold load is $33\%$.
\end{theorem}

\begin{proof}
Assume, without loss of generality, that the mean service time at each server is $1$ second. Suppose requests arrive at a rate of $\rho$ queries per second per server.

Without replication, each server evolves as an M/M/1 queue with departure rate $1$ and arrival rate $\rho$. The response time of each server is therefore exponentially distributed with rate $1-\rho$\cite{Asmussen1987}, and the mean response time is $\frac{1}{1-\rho}$.

With replication, each server is an M/M/1 queue with departure rate $1$ and arrival rate $2\rho$. The response time of each server is exponentially distributed with rate $1-2\rho$, but each query now takes the minimum of two independent\footnote{\new{Strictly speaking, the states of the queues at the servers are not truly independent of each other, but this is a reasonable approximation when the number of servers $N$ is sufficiently large compared to the level of redundancy $k$. In simulations with $k = 2$, we found that the worst-case error introduced by this approximation was $3\%$ at $N=10$ and less than $0.1\%$ at $N = 20$.}} samples from this distribution, so that the mean response time of each query is $\frac{1}{2(1 - 2\rho)}$.

Now replication results in a smaller response time if and only if $\frac{1}{2(1 - 2\rho)} < \frac{1}{1-\rho}$, i.e., when $\rho < \frac{1}{3}$.
\end{proof}

In the general service time case, two natural (service-time independent) bounds on the threshold load exist.

\begin{figure*}[ht]
\centering
 \subfigure[Weibull]{
   \includegraphics[width=0.3\textwidth]{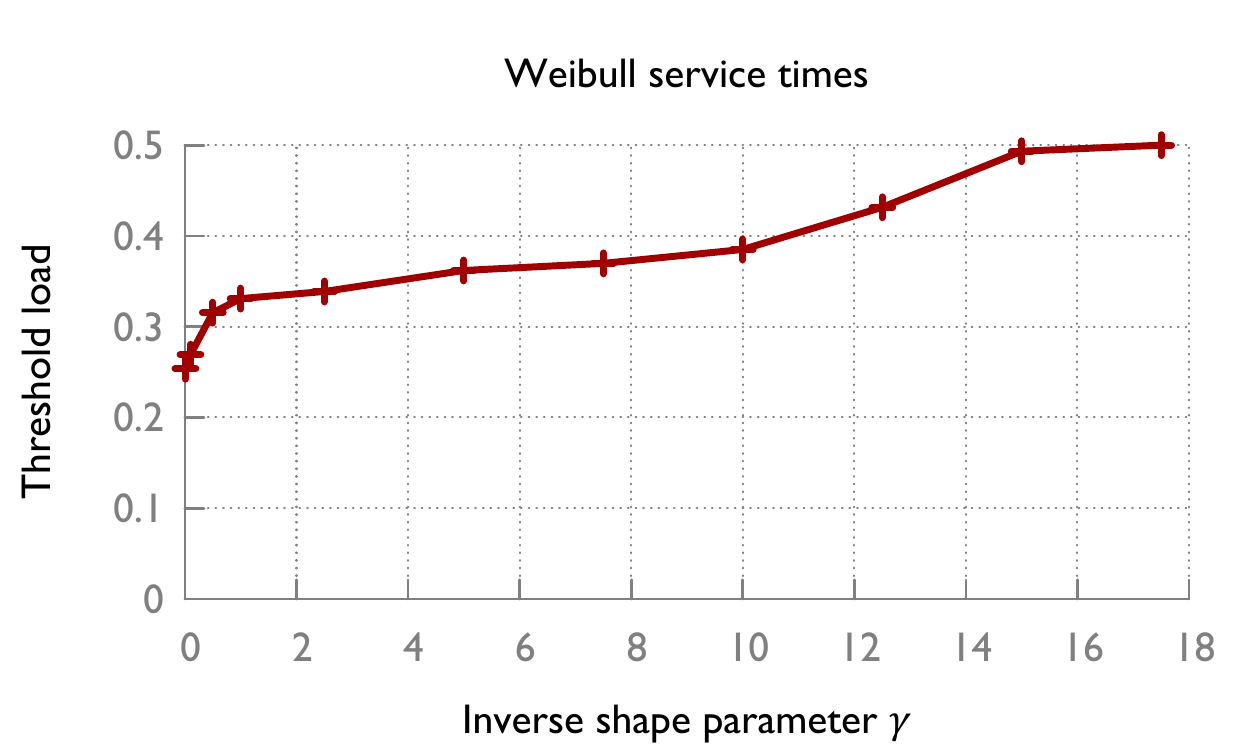}
   \label{fig:weibull}
 }
 \subfigure[Pareto]{
   \includegraphics[width=0.3\textwidth]{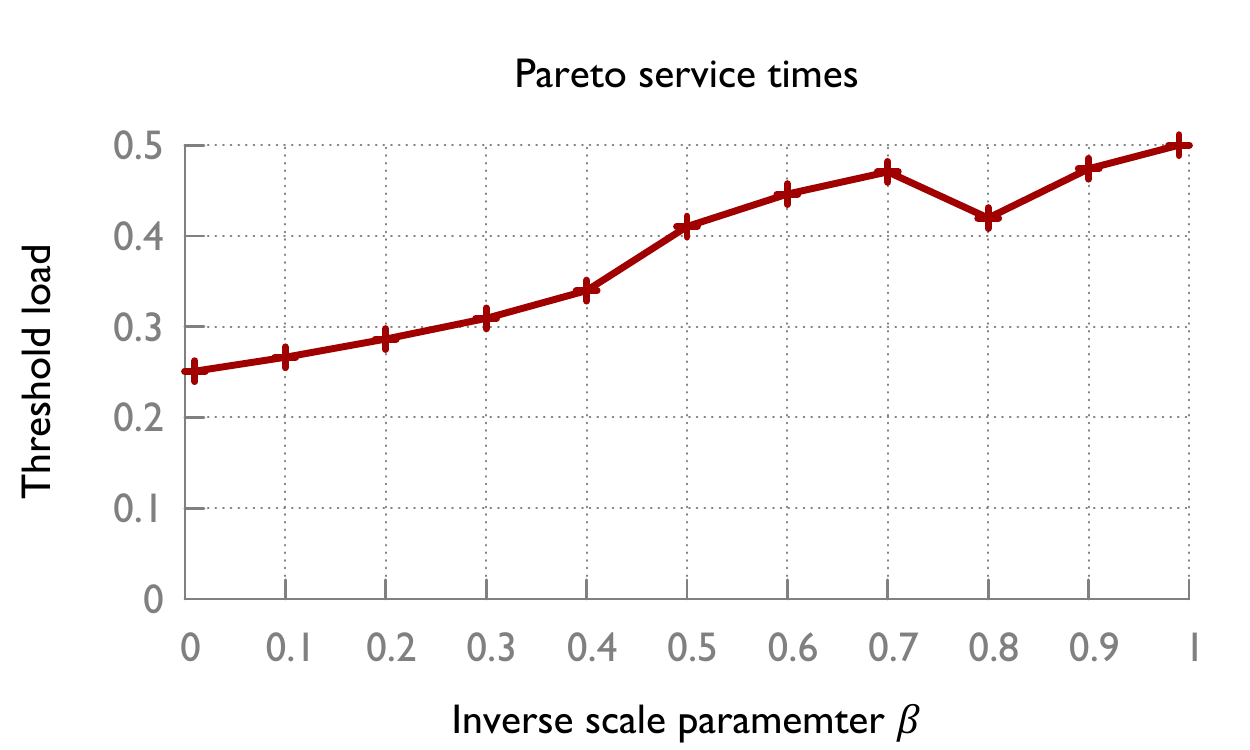}
   \label{fig:pareto-service}
 }
 \subfigure[Two-point discrete distribution]{
   \includegraphics[width=0.3\textwidth]{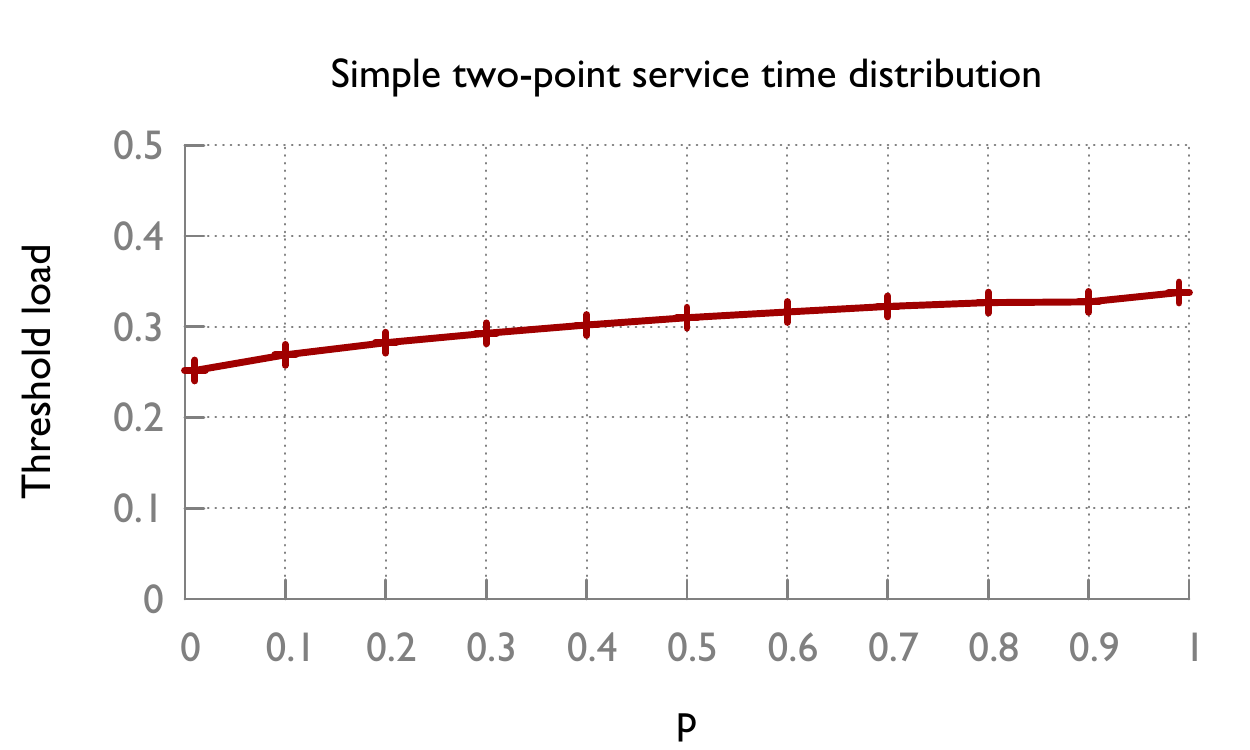}
   \label{fig:two-point}
 }
 \caption{Effect of increasing variance on the threshold load in three families of unit-mean distributions: Pareto, Weibull, and a simple two-point discrete distribution (service time = $0.5$ with probability $p$, $\frac{1-0.5p}{1-p}$ with probability $1-p$).  In all three cases the variance is $0$ at $x = 0$ and increases along the x-axis, going to infinity at the right edge of the plot.}
 \label{fig:queue-examples}
\end{figure*}
%
%
%
 
First, the threshold load cannot exceed $50\%$ load in any system. This is easy to see: if the base load is above $50\%$, replication would push total load above $100\%$. It turns out that this trivial upper bound is tight --- there are families of heavy-tailed high-variance service times for which the threshold load goes arbitrarily close to $50\%$. See Figures~\ref{fig:weibull} and \ref{fig:pareto-service}.

Second, we intuitively expect replication to help more as the service time distribution becomes more variable. Figure~\ref{fig:queue-examples} validates this trend in three different families of distributions. Therefore, it is reasonable to expect that the worst-case for replication is when the service time is completely deterministic. However, even in this case the threshold load is strictly positive because there is still variability in the system due to the stochastic nature of the arrival process. With the Poisson arrivals that we assume, the threshold load with deterministic service time turns out to be slightly less than $26\%$ --- more precisely, $\approx 25.82\%$ --- based on simulations of the queueing model, as shown in the leftmost point in Figure~\ref{fig:two-point}.

We conjecture that this is, in fact, a lower bound on the threshold load in an arbitrary system.

\begin{conjecture}
\label{conjecture:queueing}
Deterministic service time is the worst case for replication: there is no service time distribution in which the threshold load is below the ($\approx 26\%$) threshold when the service time is deterministic.
\end{conjecture}

The primary difficulty in resolving the conjecture is that general response time distributions are hard to handle analytically, especially since in order to quantify the effect of taking the minimum of two samples we need to understand the shape of the entire distribution, not just its first few moments. However, we have two forms of evidence that seem to support this conjecture: analyses based on approximations to the response time distribution, and simulations of the queueing model.

The primary approximation that we use is a recent result by Myers and Vernon~\cite{MyersVernon2012} that only depends on the first two moments of the service time distribution. The approximation seems to perform fairly well in numerical evaluations with light-tailed service time distributions, such as the Erlang and hyperexponential distributions (see Figure 2 in \cite{MyersVernon2012}), although no bounds on the approximation error are available. However, the authors note that the approximation is likely to be inappropriate when the service times are heavy tailed.

As a supplement, therefore, in the heavy-tailed case, we use an approximation by Olvera-Cravioto et al.~\cite{Olvera-Cravioto2011} that is applicable when the service times are regularly varying\footnote{The class of regularly varying distributions is an important subset of the class of heavy-tailed distributions that includes as its members the Pareto and the log-Gamma distributions.}. Heavy-tail approximations are fairly well established in queueing theory (see~\cite{Asmussen1987}); the result due to Olvera-Cravioto et al.\ is, to the best of our knowledge, the most recent (and most accurate) refinement.

The following theorems summarize our results for these approximations.
\new{We omit the proofs due to space constraints.}

\begin{theorem}
\label{thm:myers-vernon}
Within the approximation due to Myers and Vernon~\cite{MyersVernon2012} of the response time distribution, the threshold load is minimized when the service time distribution is deterministic.
\end{theorem}

The heavy-tail approximation by Olvera-Cravioto et al.~\cite{Olvera-Cravioto2011} applies to arbitrary regularly varying service time distributions, but for our analysis we add an additional assumption requiring that the service time be sufficiently heavy. Formally, we require that the service time distribution have a higher coefficient of variation than the exponential distribution, which amounts to requiring that the tail index $\alpha$ be $< 1 + \sqrt{2}$. (The tail index is a measure of how heavy a distribution is: lower indices mean heavier tails.)

\begin{theorem}
\label{thm:olvera-cravioto}
Within the approximation due to Olvera-Cravioto et al.~\cite{Olvera-Cravioto2011}, if the service time distribution is regularly varying with tail index $\alpha < 1 + \sqrt{2}$, then the threshold load is $> 30\%$.
\end{theorem}

\begin{figure}
\centering
\includegraphics[width=0.45\textwidth]{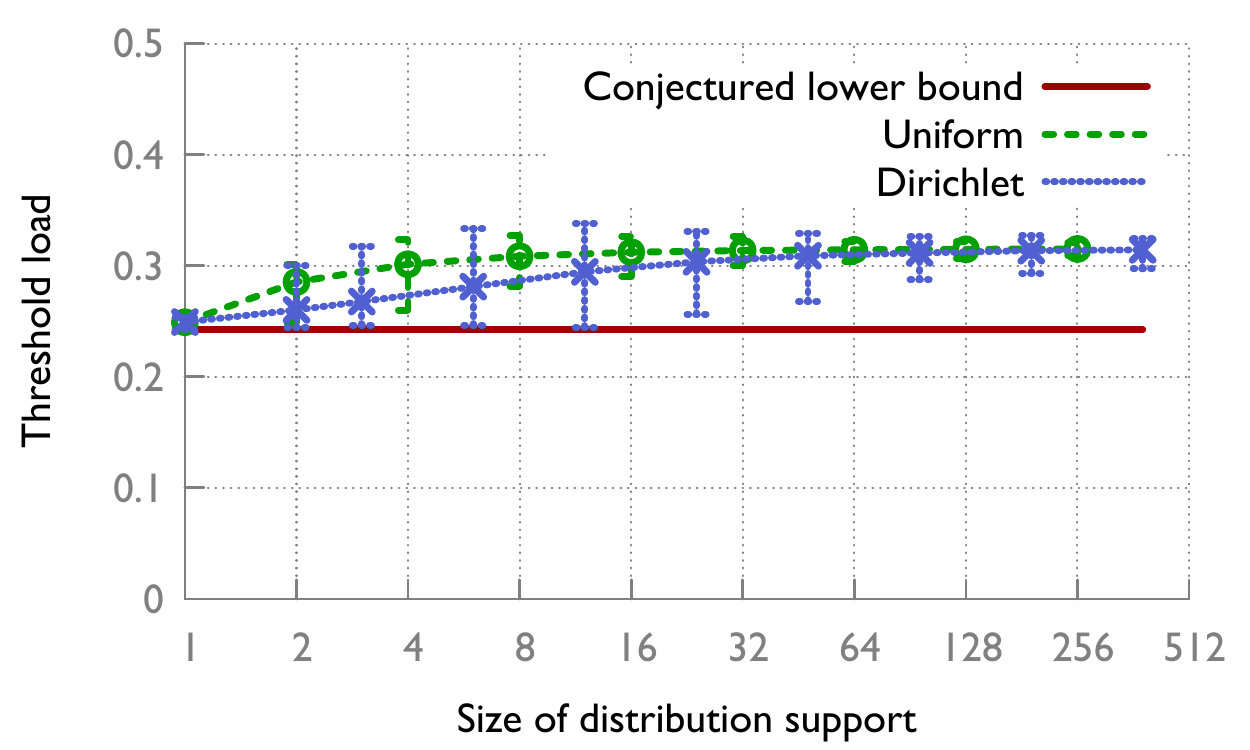}
\caption{Randomly chosen service time distributions}
\label{fig:random-discrete}
\end{figure}

Simulation results also seem to support the conjecture. We generated a range of service time distributions by, for various values of $N$, sampling from the space of all unit-mean discrete probability distributions with support $\{1, 2, ..., N\}$ in two different ways --- uniformly at random, and using a symmetric Dirichlet distribution with concentration parameter $0.1$ (the Dirichlet distribution has a higher variance and generates a larger spread of distributions than uniform sampling). Figure~\ref{fig:random-discrete} reports results when we generate a $1000$ different random distributions for each value of $N$ and look at the minimum and maximum observed threshold load over this set of samples.

\subsubsection*{Effect of client-side overhead}
 
\begin{figure}
\includegraphics[width=0.45\textwidth]{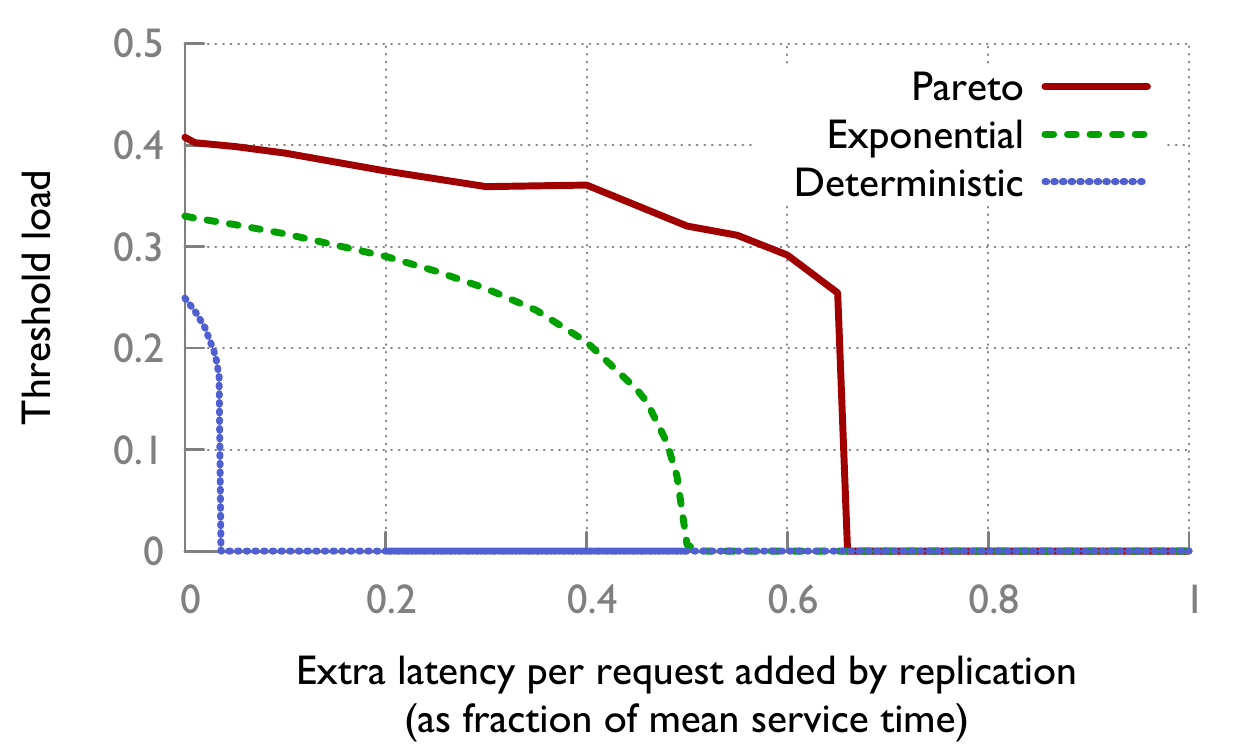}
\caption{Effect of redundancy-induced client-side latency overhead, with different server service time distributions.}
\label{fig:client-overhead}
\end{figure}
 
As we noted earlier, our analysis so far assumes that the client-side overhead (e.g.\ added CPU utilization, kernel processing, network overhead) involved in processing the replicated requests is negligible. This may not be the case when, for instance, the operations in question involve large file transfers or very quick memory accesses. In both cases, the client-side latency overhead involved in processing an additional replicated copy of a request would be comparable in magnitude to the server latency for processing the request. This overhead can partially or completely counteract the latency improvement due to redundancy. Figure~\ref{fig:client-overhead} quantifies this effect by considering what happens when replication adds a fixed latency penalty to every request. These results indicate that the more variable distributions are more forgiving of overhead, but client side overhead must be at least somewhat smaller than mean request latency in order for replication to improve \emph{mean} latency.  This is not surprising, of course: if replication overhead equals mean latency, replication cannot improve mean latency for any service time distribution --- though it may still improve the tail.
\subsection{Application: disk-backed database}
\label{sec:apache}

Many data center applications involve the use of a large disk-based
data store that is accessed via a smaller main-memory cache: examples
include the Google AppEngine data store~\cite{GAEMemcached}, Apache
Cassandra~\cite{Cassandra}, and Facebook's Haystack image
store~\cite{Haystack}. In this section we study a representative
implementation of such a storage service: a set of Apache web servers
hosting a large collection of files, split across the servers via
consistent hashing, with the Linux kernel managing a disk cache on
each server.

We deploy a set of Apache servers and, using a light-weight
memory-soaking process, adjust the memory usage on each server node so
that around half the main memory is available for the Linux disk cache
(the other half being used by other applications and the kernel). We
then populate the servers with a collection of files whose total size
is chosen to achieve a preset target cache-to-disk ratio. The files
are partitioned across servers via consistent hashing, and two copies
are stored of every file: if the primary is stored on server $n$, the
(replicated) secondary goes to server $n+1$. We measure the response
time when a set of client nodes generate requests according to
identical Poisson processes. Each request downloads a file chosen
uniformly at random from the entire collection. We only test read
performance on a static data set; we do not consider writes or
updates.

\begin{figure*}[!ht]
  \centering
  \subfigure{
  \includegraphics[width=0.3\textwidth]{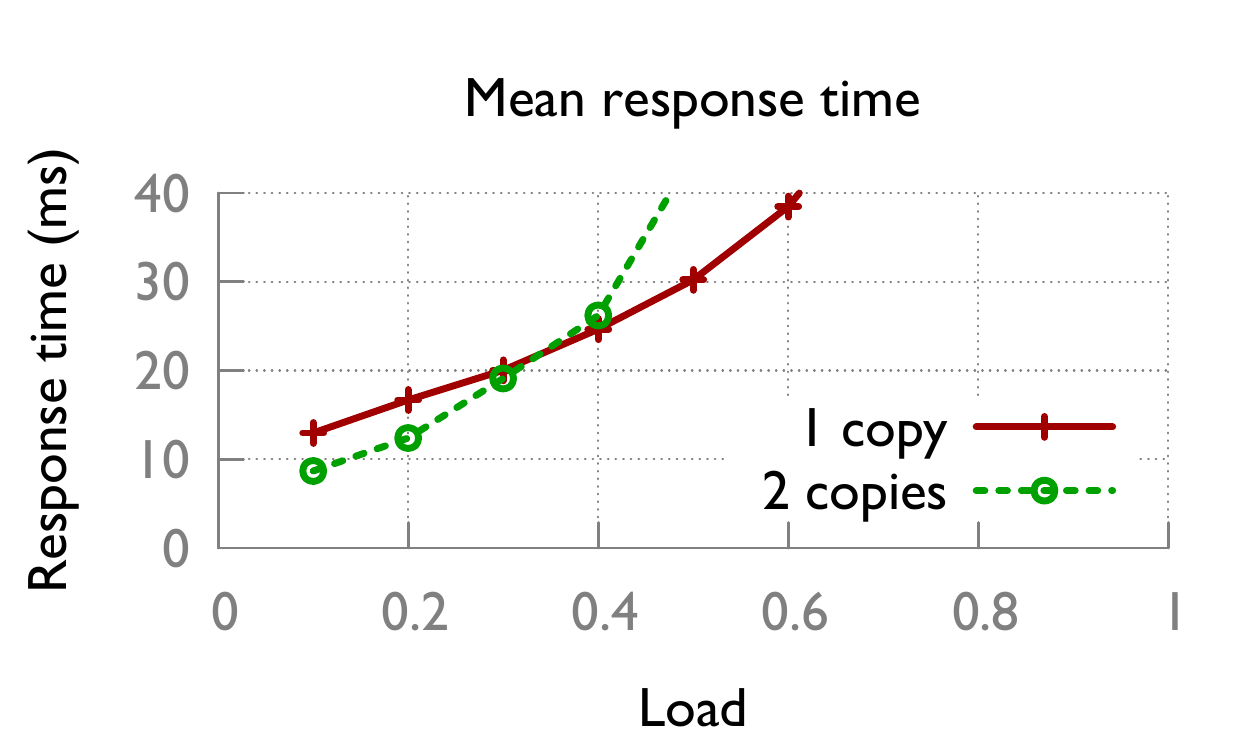}
  }
  \subfigure{
  \includegraphics[width=0.3\textwidth]{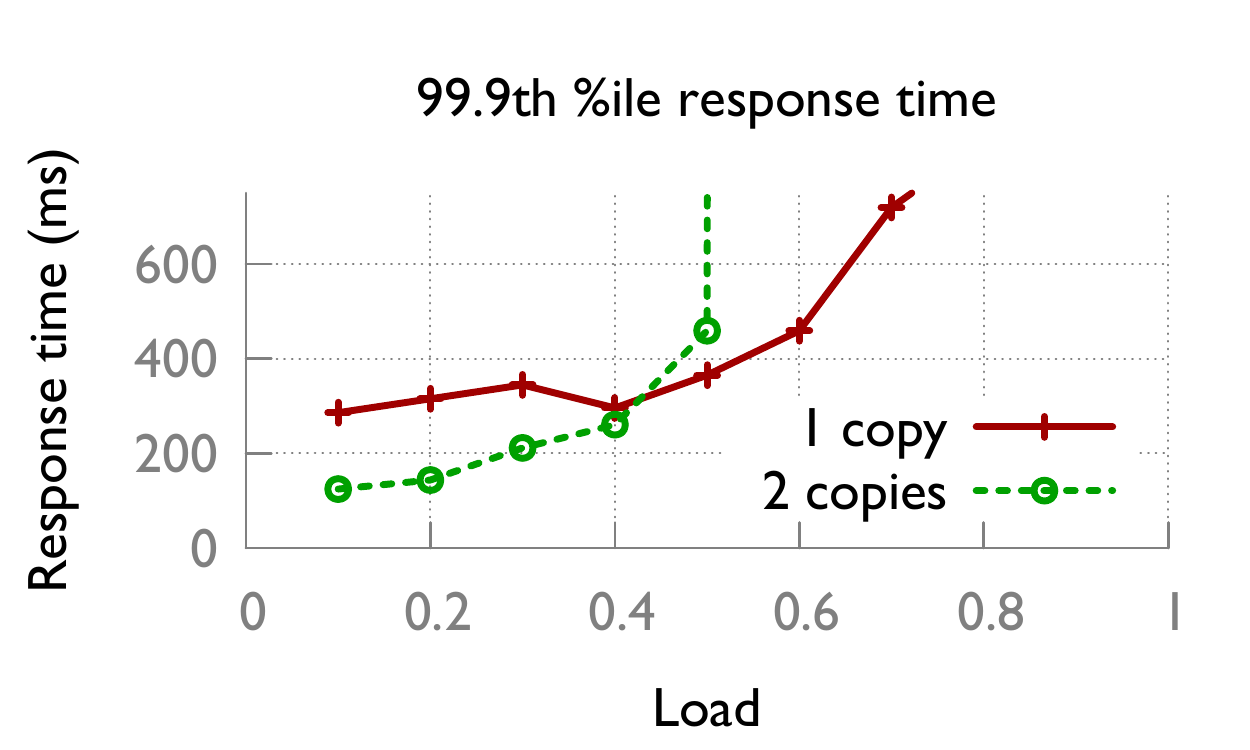}
  }
  \subfigure{
  \includegraphics[width=0.3\textwidth]{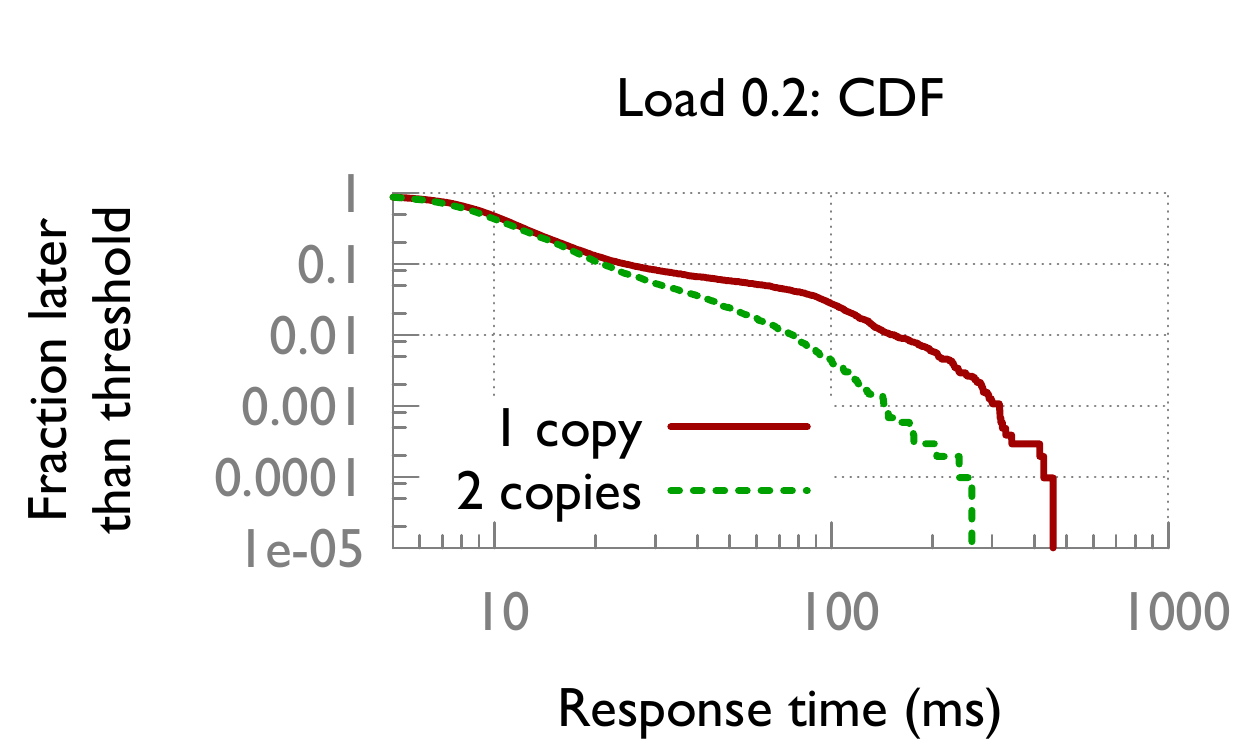}
  }
  \caption{Base configuration}
  \label{fig:basecase}
\end{figure*}

\begin{figure*}
  \centering
  \subfigure{
  \includegraphics[width=0.3\textwidth]{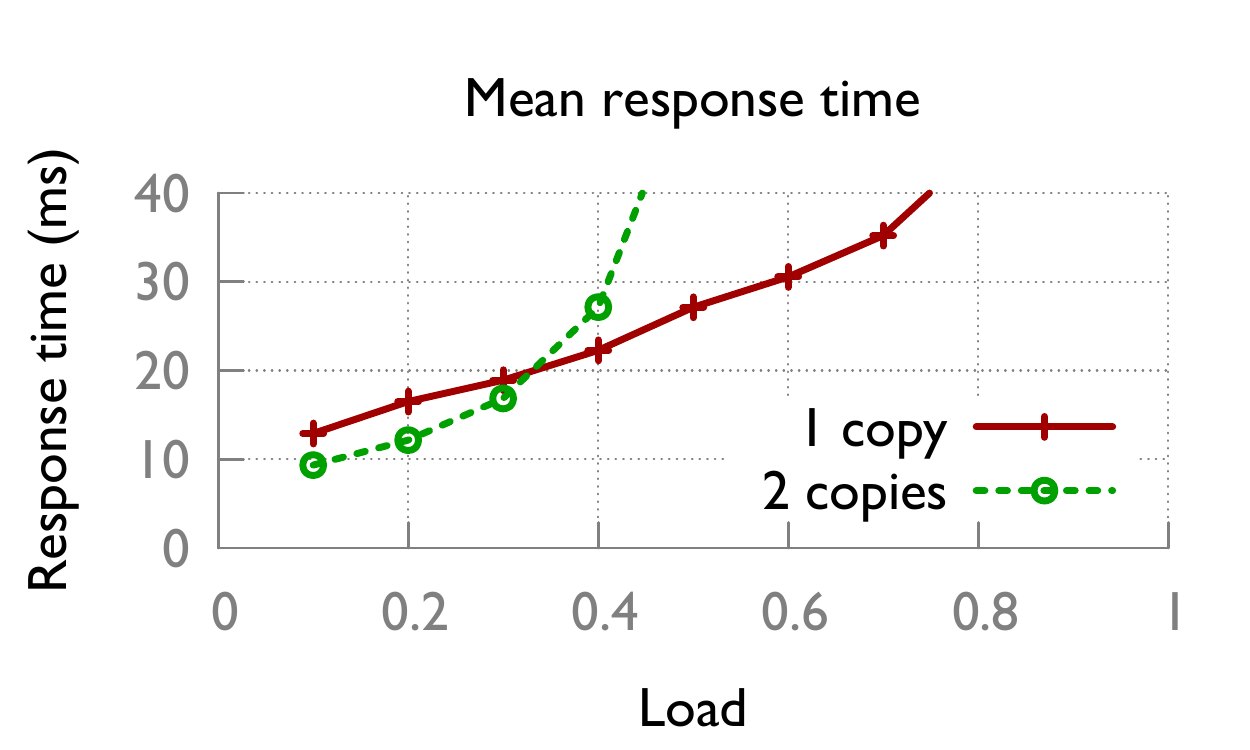}
  }
  \subfigure{
  \includegraphics[width=0.3\textwidth]{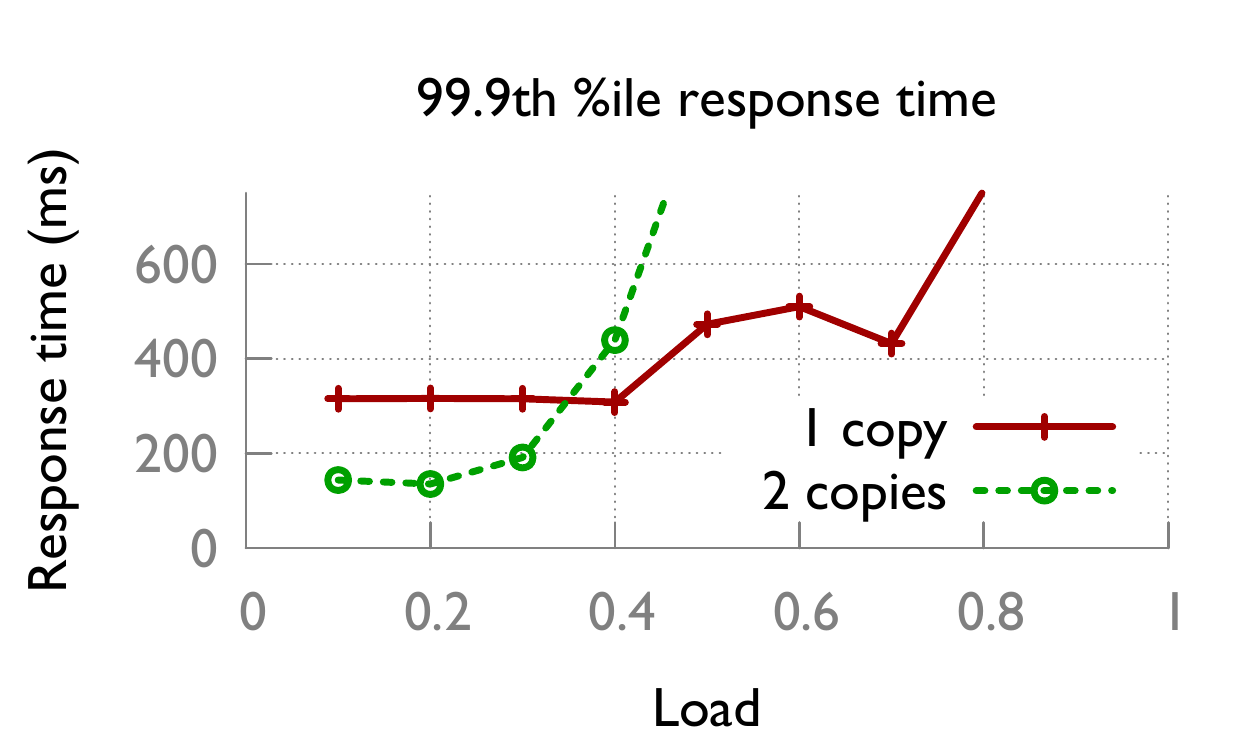}
  }
  \subfigure{
  \includegraphics[width=0.3\textwidth]{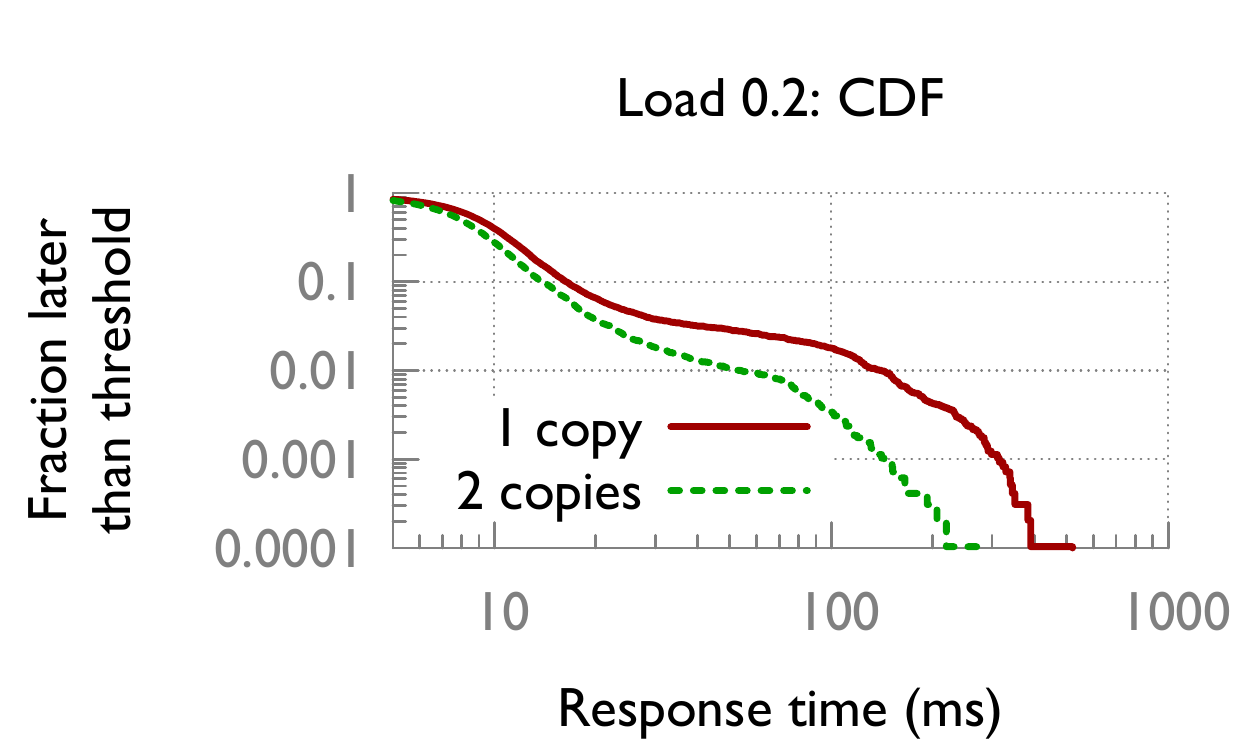}
  }
  \caption{Mean file size $0.04$ KB instead of $4$ KB}
  \label{fig:filesize41B}
\end{figure*}

\begin{figure*}
  \centering
  \subfigure{
  \includegraphics[width=0.3\textwidth]{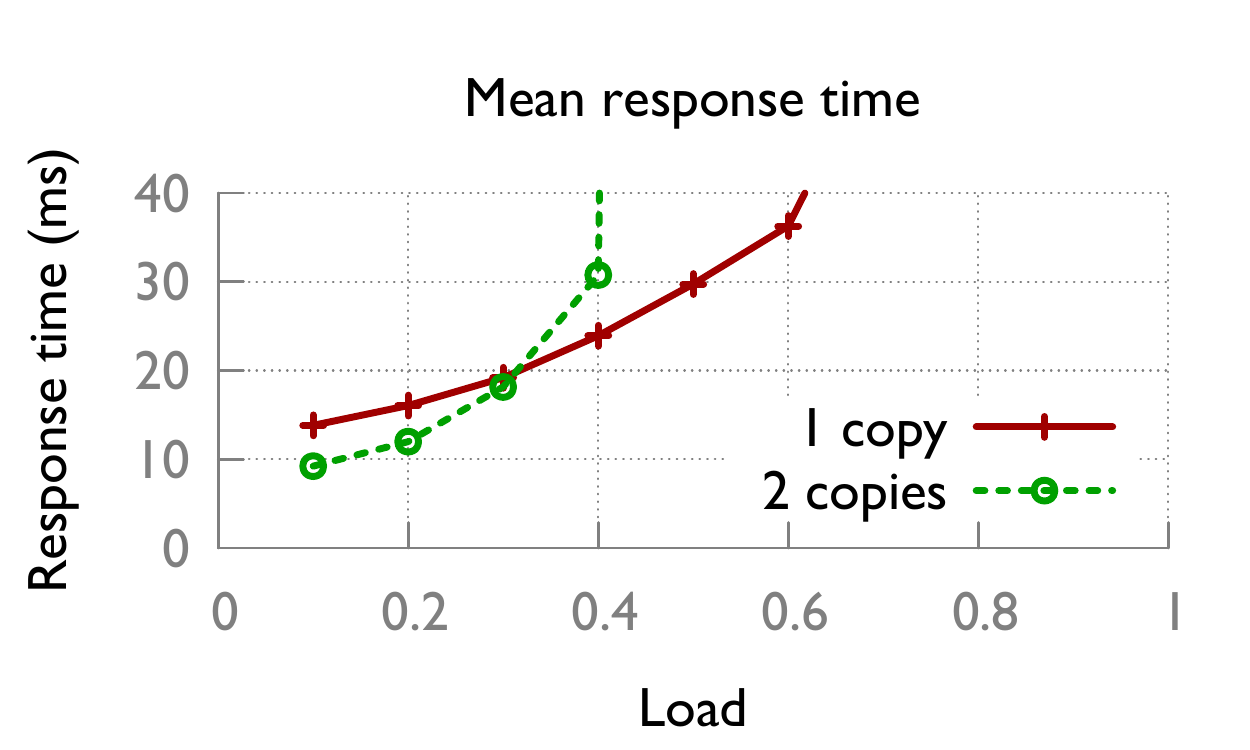}
  }
  \subfigure{
  \includegraphics[width=0.3\textwidth]{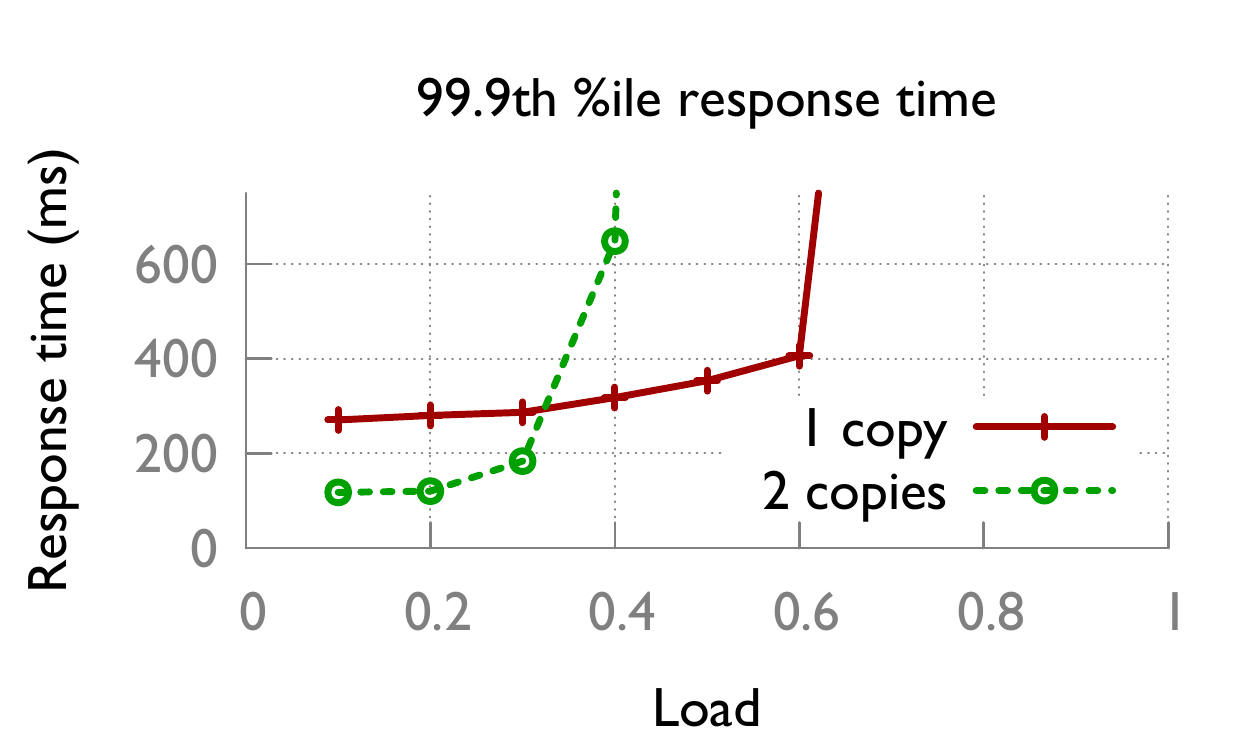}
  }
  \subfigure{
  \includegraphics[width=0.3\textwidth]{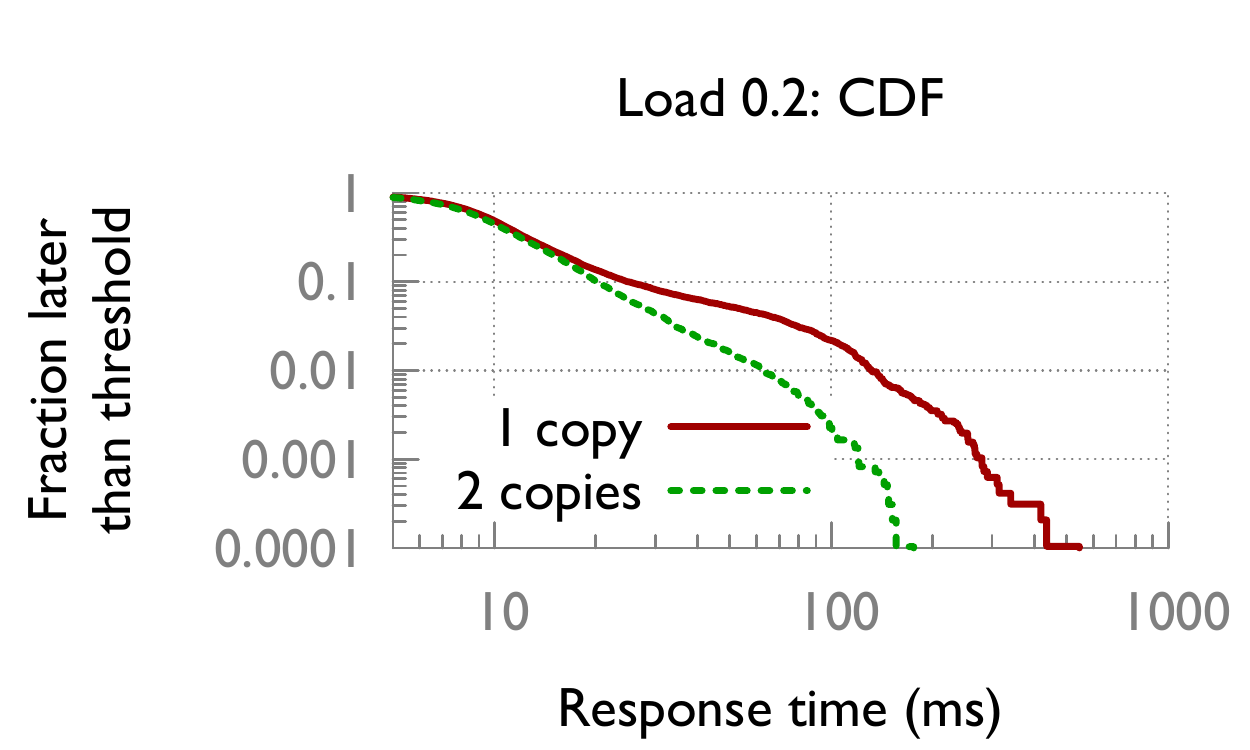}
  }
\caption{Pareto file size distribution instead of deterministic}
\label{fig:pareto}
\end{figure*}

\begin{figure*}
\centering
  \subfigure{
  \includegraphics[width=0.3\textwidth]{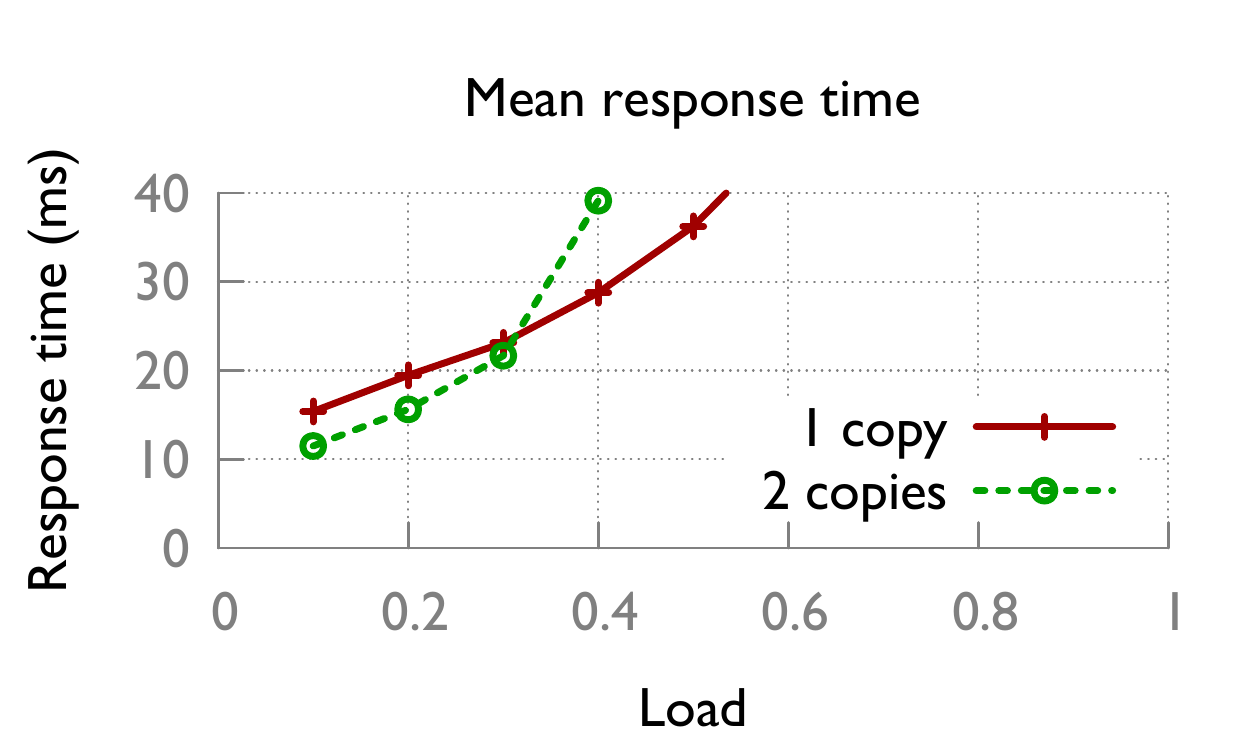}
  }
  \subfigure{
  \includegraphics[width=0.3\textwidth]{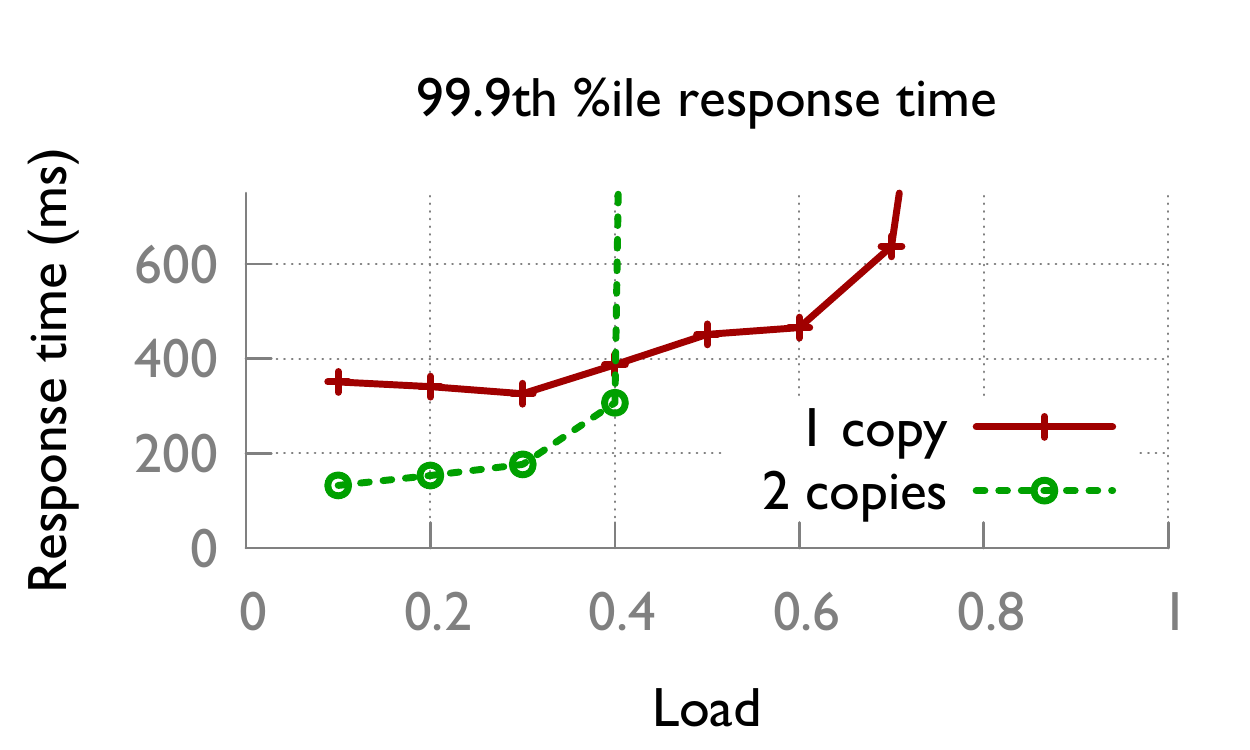}
  }
  \subfigure{
  \includegraphics[width=0.3\textwidth]{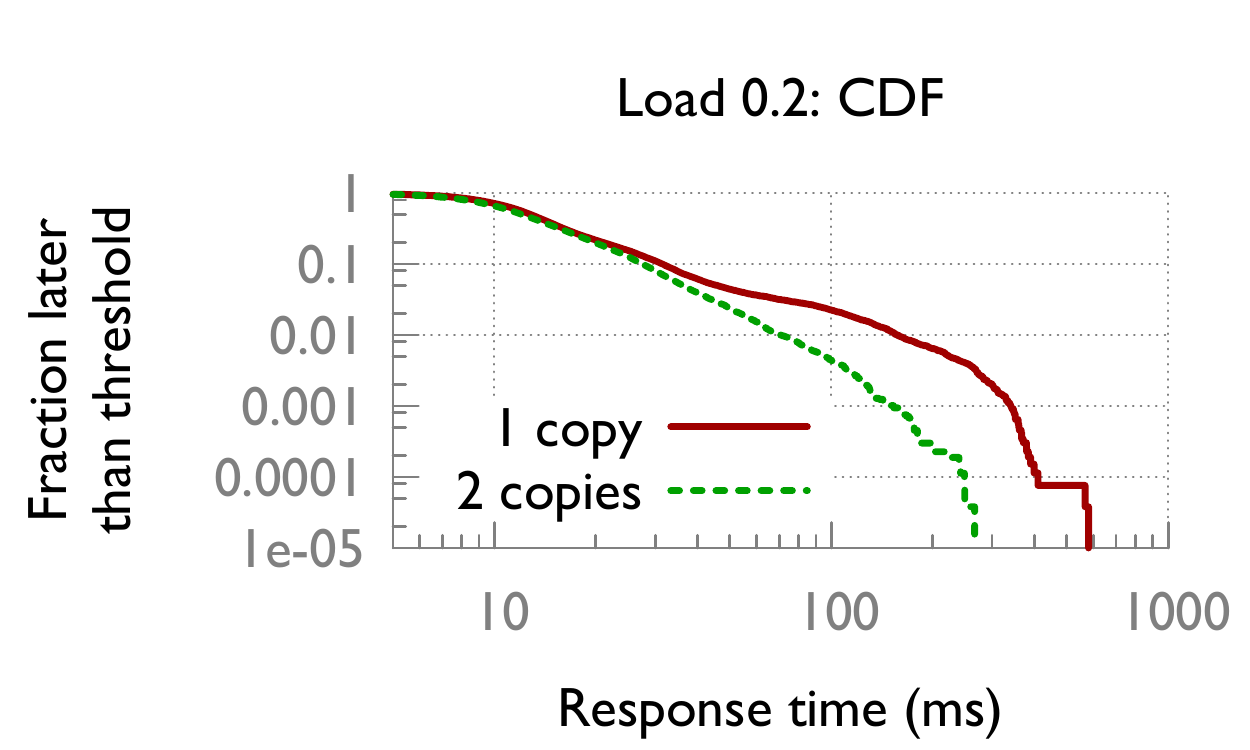}
  }
  \caption{Cache:disk ratio $0.01$ instead of $0.1$. Higher
    variability because of the larger proportion of accesses hitting
    disk. Compared to Figure~\ref{fig:basecase}, $99.9$th percentile
    improvement goes from $2.3\times$ to $2.8\times$ at $10\%$ load,
    and from $2.2\times$ to $2.5\times$ at $20\%$ load.}
\label{fig:ctd1t100}
\end{figure*}

\begin{figure*}
\centering
  \subfigure{
  \includegraphics[width=0.3\textwidth]{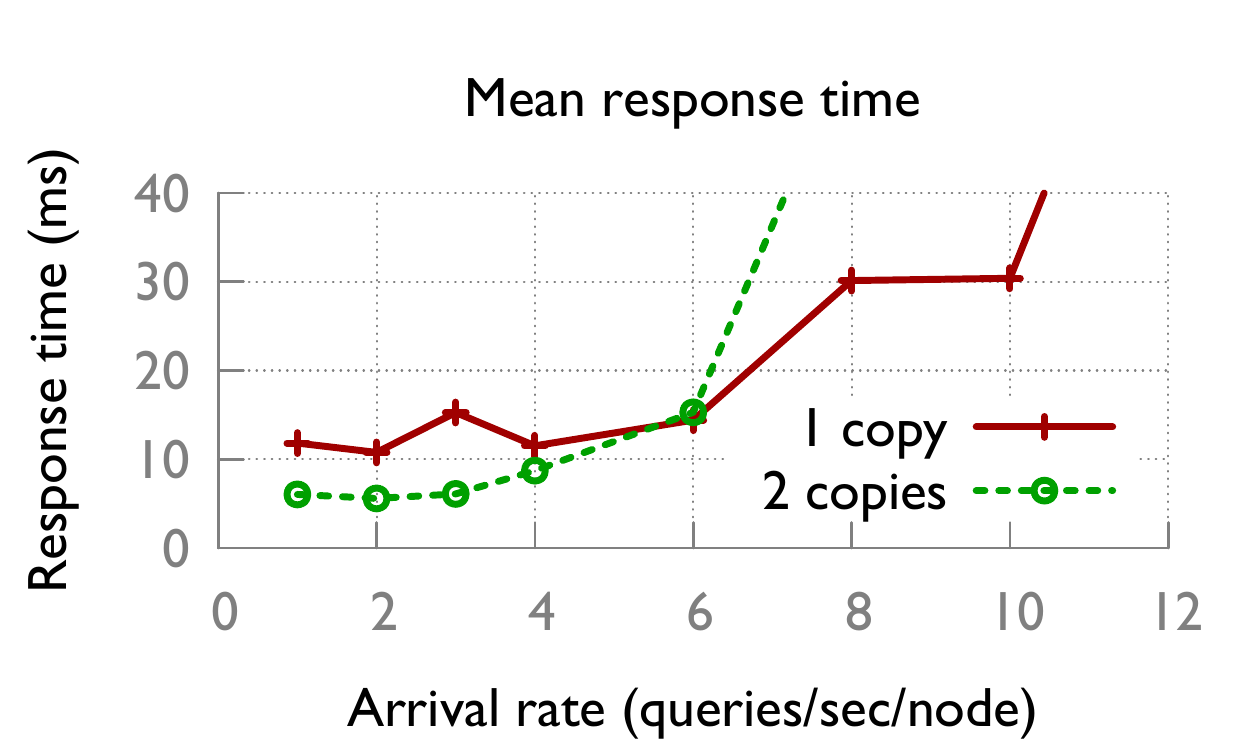}
  }
  \subfigure{
  \includegraphics[width=0.3\textwidth]{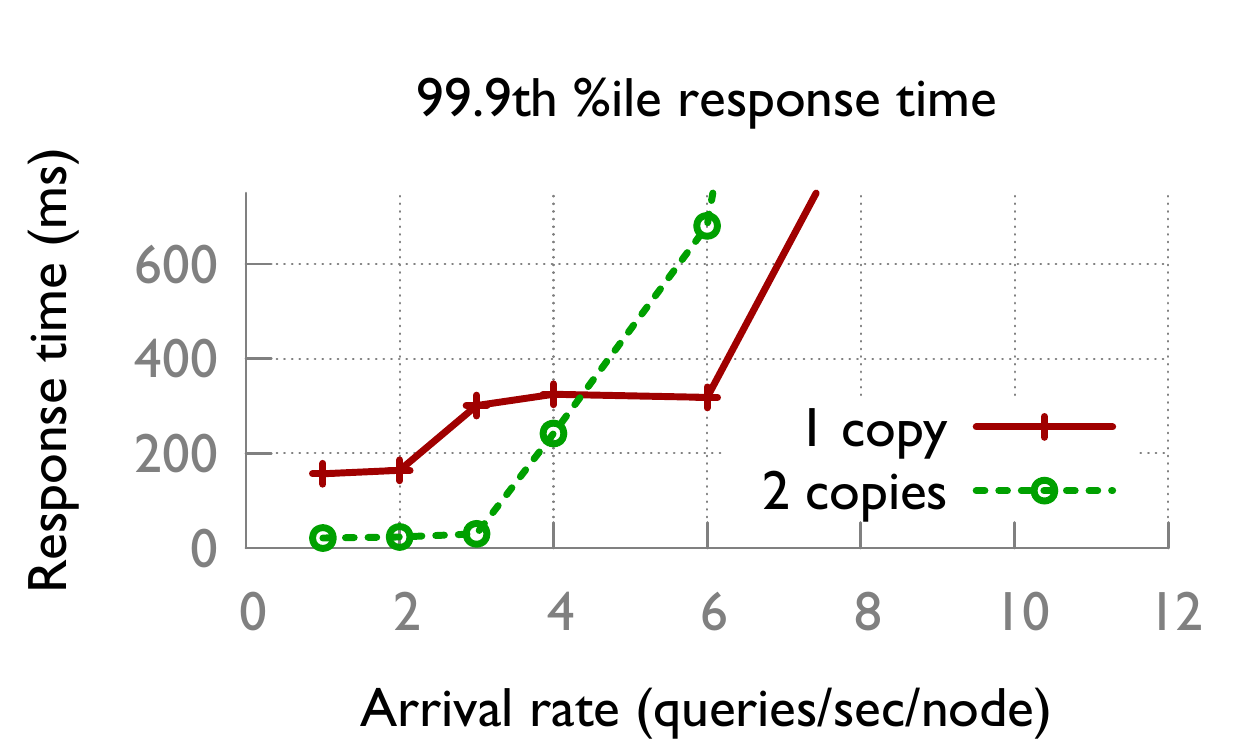}
  }
  \subfigure{
  \includegraphics[width=0.3\textwidth]{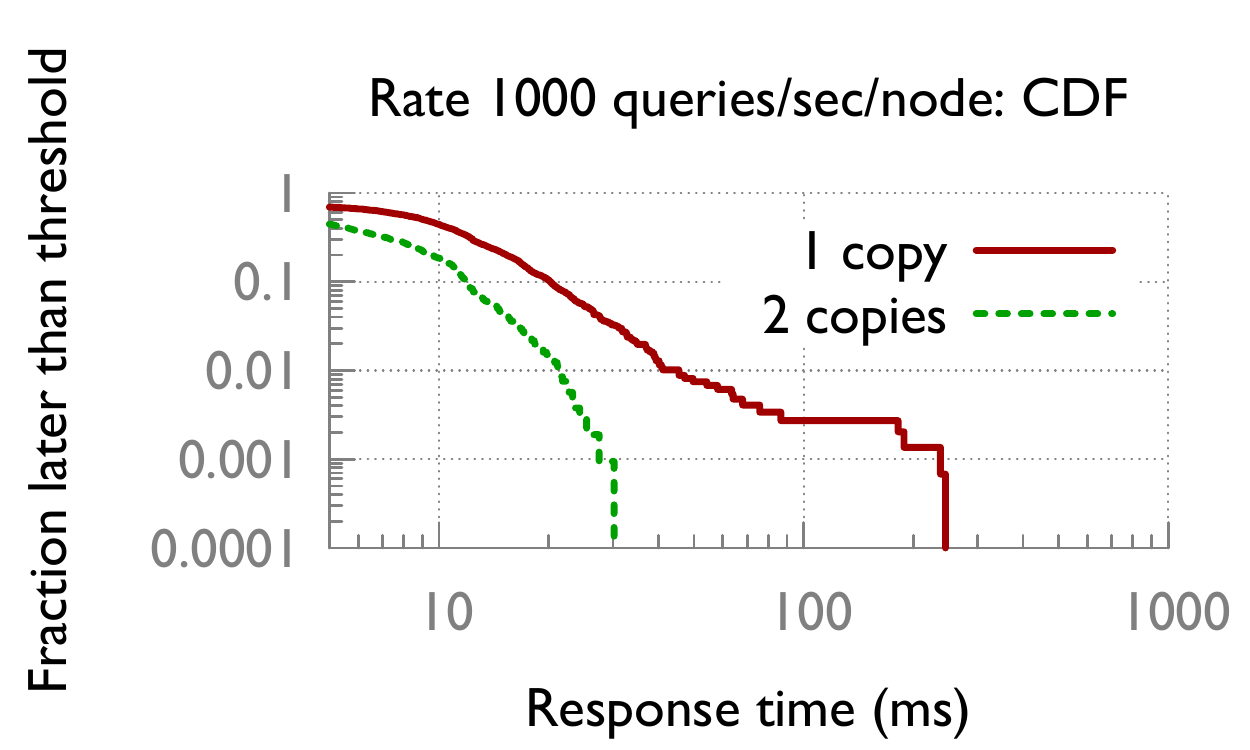}
  }
\caption{EC2 nodes instead of Emulab. $x$-axis shows unnormalised arrival
  rate because maximum throughput seems to fluctuate. Note the much
  larger tail improvement compared to Figure~\ref{fig:basecase}.}
\label{fig:ec2}
\end{figure*}

\begin{figure*}
  \centering
  \subfigure{
  \includegraphics[width=0.3\textwidth]{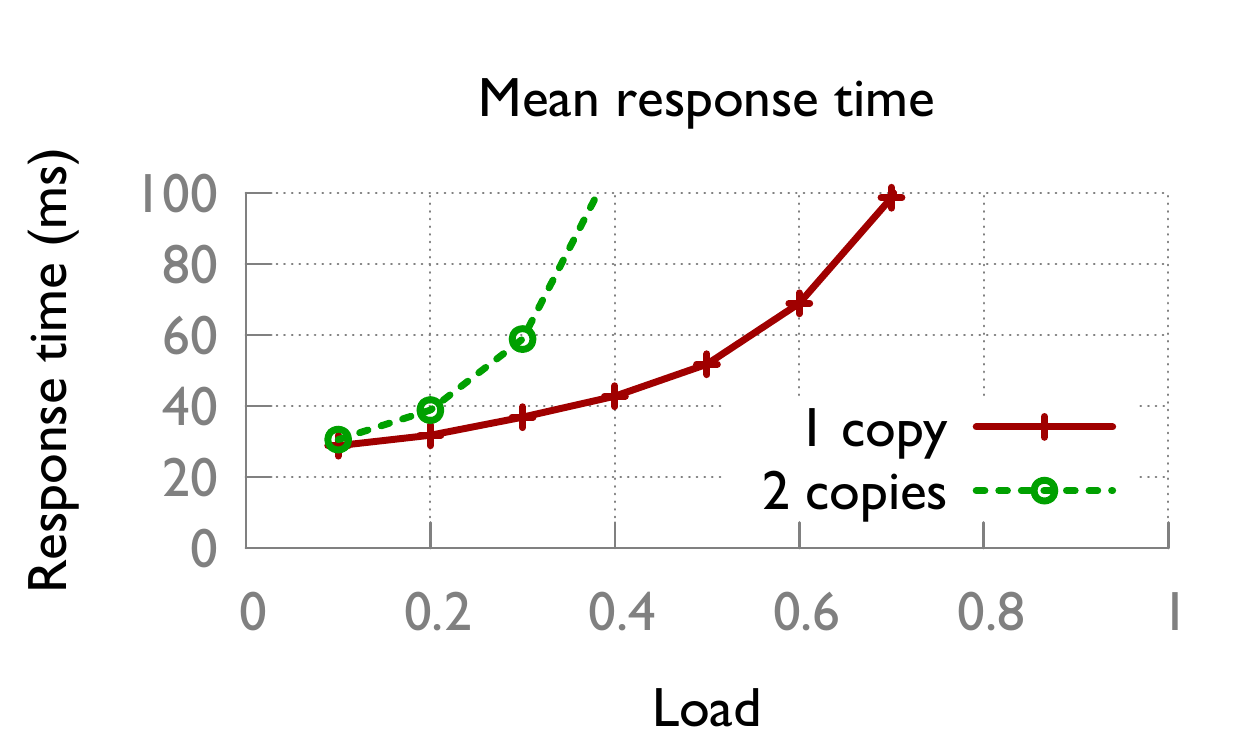}
  }
  \subfigure{
  \includegraphics[width=0.3\textwidth]{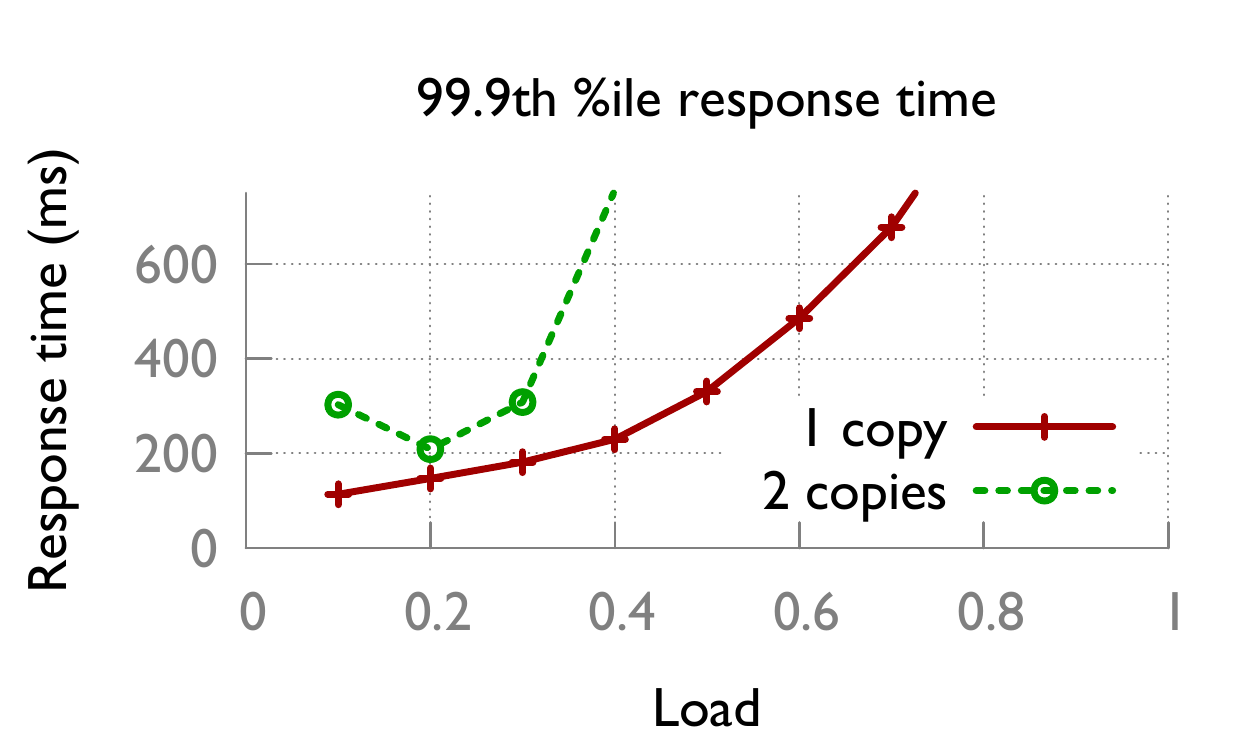}
  }
  \subfigure{
  \includegraphics[width=0.3\textwidth]{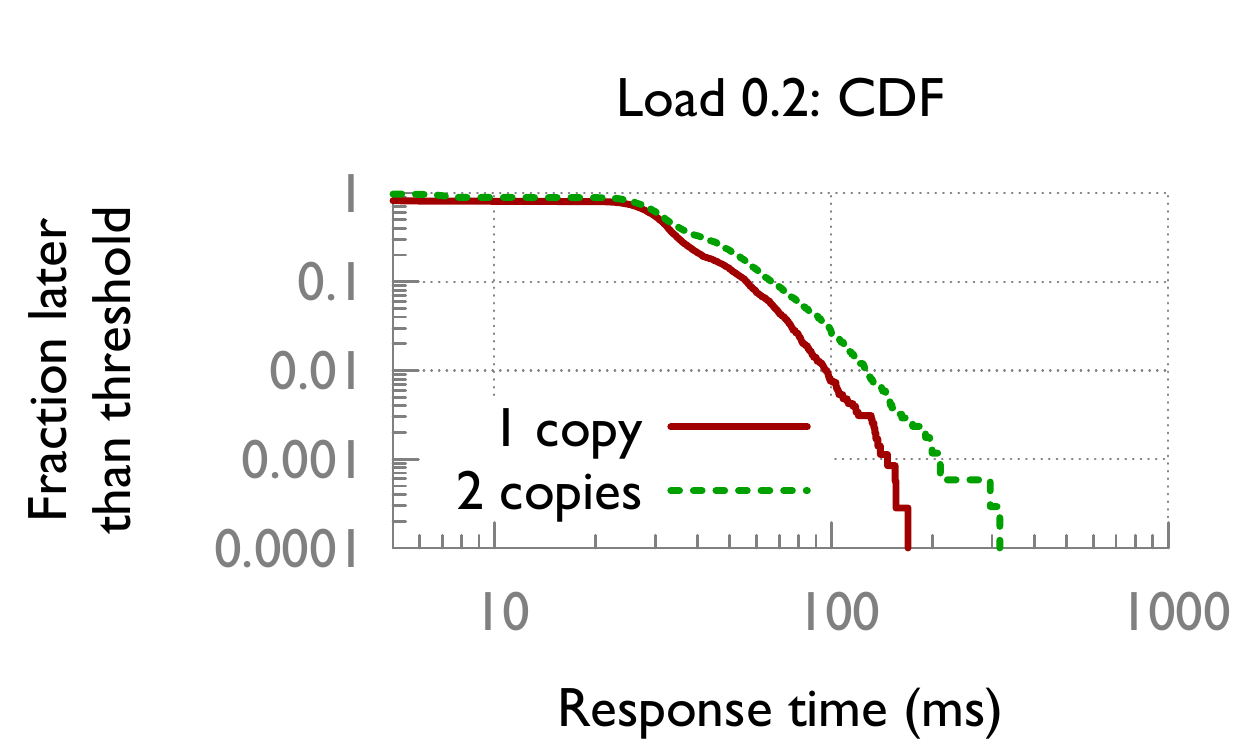}
  }
  \caption{Mean file size $400$ KB instead of $4$ KB}
  \label{fig:filesize400KB}
\end{figure*}

\begin{figure*}
\centering
  \subfigure{
  \includegraphics[width=0.3\textwidth]{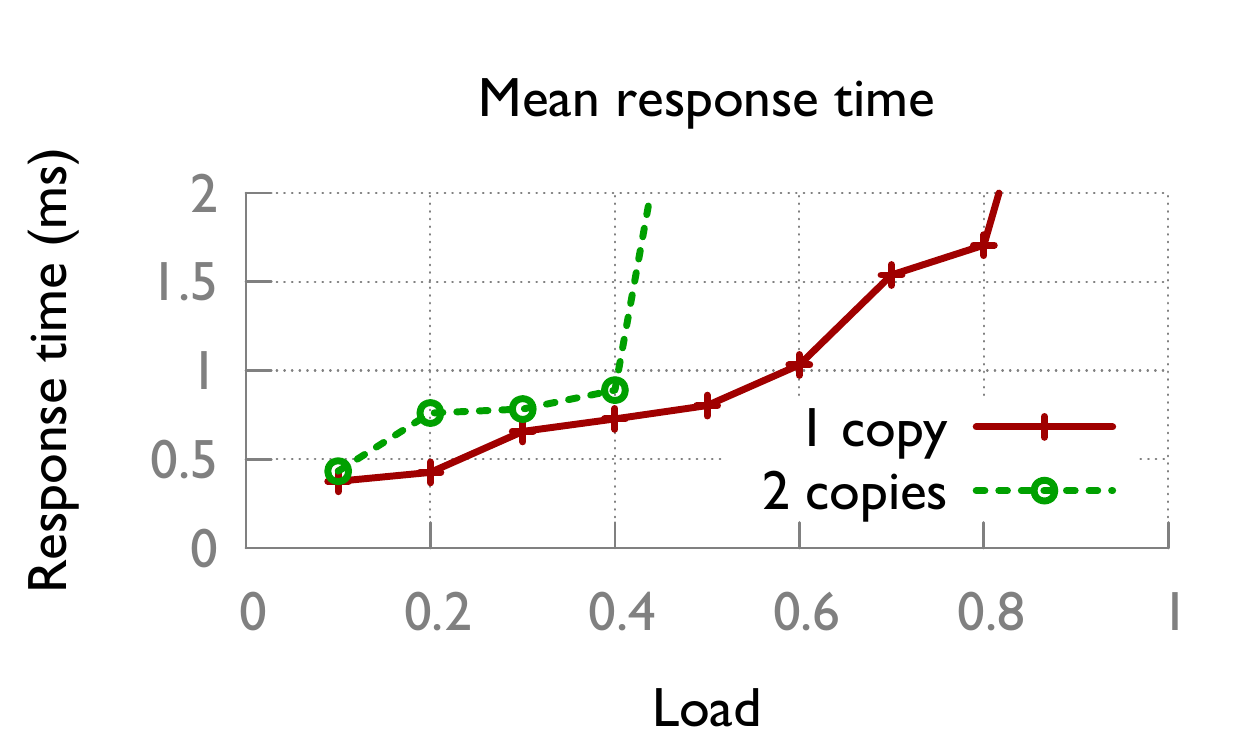}
  }
  \subfigure{
  \includegraphics[width=0.3\textwidth]{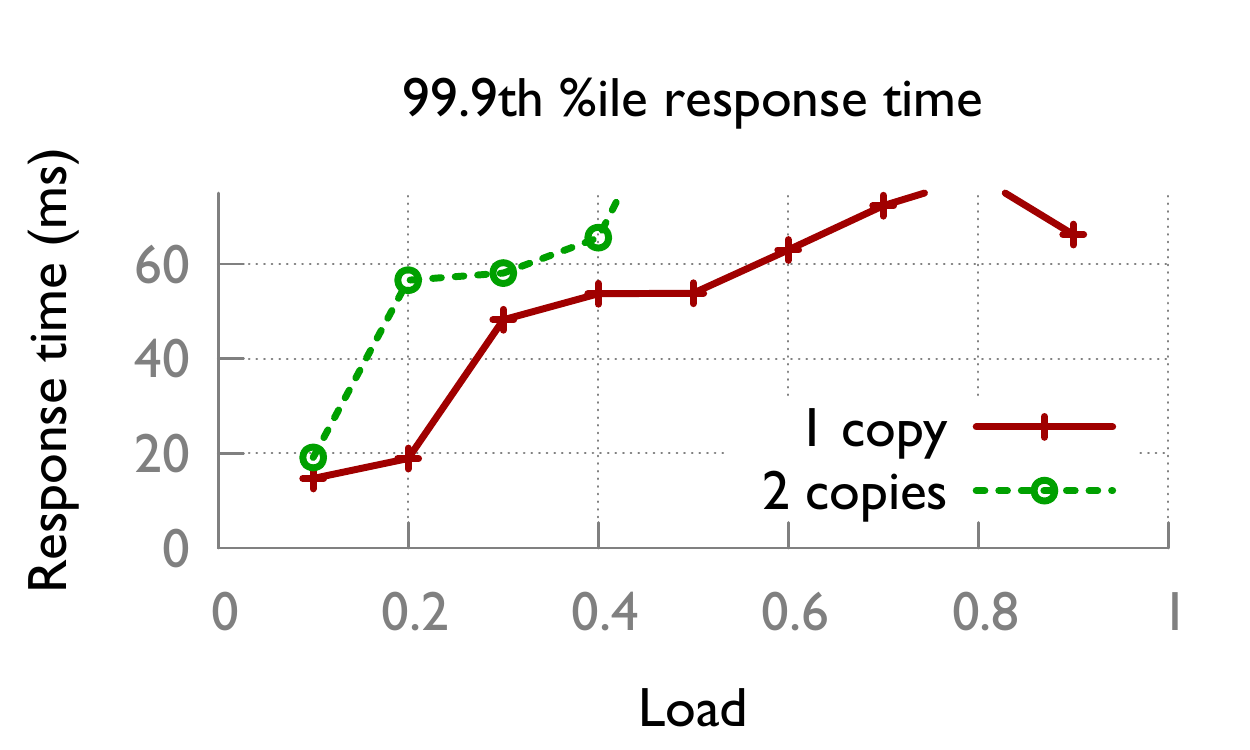}
  }
  \subfigure{
  \includegraphics[width=0.3\textwidth]{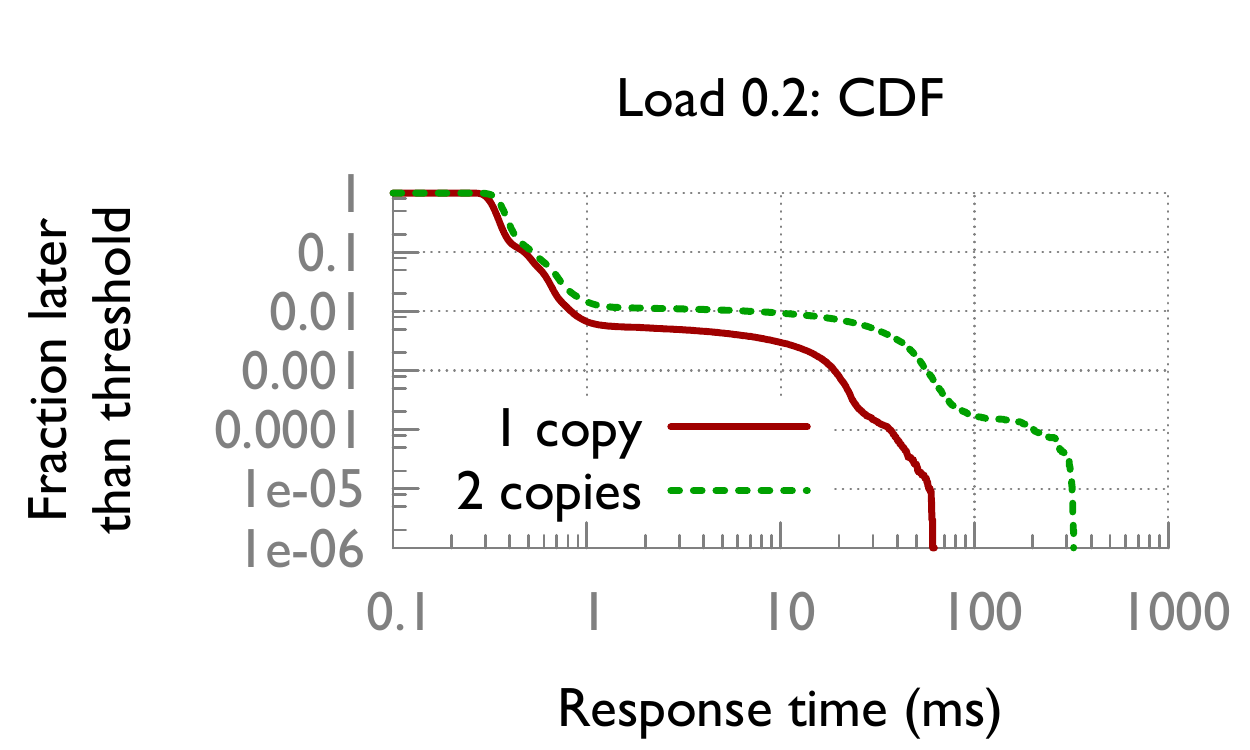}
  }
\caption{Cache:disk ratio $2$ instead of $0.1$. Cache is large enough to
  store contents of entire disk}
\label{fig:ctd2t1}
\end{figure*}

Figure~\ref{fig:basecase} shows results for one particular web-server
configuration, with
\begin{itemize*}
  \item Mean file size = $4$ KB
  \item File size distribution = deterministic, $4$ KB per file
  \item Cache:disk ratio = $0.1$
  \item Server/client hardware = 4 servers and 10 clients, all
    identical single-core Emulab nodes with 3 GHz CPU, 2 GB RAM,
    gigabit network interfaces, and 10k RPM disks.
\end{itemize*}
Disk is the
bottleneck in the majority of our experiments -- CPU and network
usage are always well below peak capacity.

The threshold load (the maximum load below which replication always
helps) is $30\%$ in this setup --- within the $25$-$50\%$ range
predicted by the queueing analysis. Redundancy reduces mean latency by
$33\%$ at $10\%$ load and by $25\%$ at $20\%$ load. Most of the
improvement comes from the tail. At $20\%$ load, for instance,
replication cuts $99$th percentile latency in half, from $150$ ms to
$75$ ms, and reduces $99.9$th percentile latency $2.2\times$.

 
The experiments in subsequent figures (Figures~\ref{fig:filesize41B}-\ref{fig:ctd2t1}) vary one of the above configuration
parameters at a time, keeping the others fixed. We note three
observations.

First, as long as we ensure that file sizes continue to remain
relatively small, changing the mean file size
(Figure~\ref{fig:filesize41B}) or the shape of the file size
distribution (Figure~\ref{fig:pareto}) does not siginificantly alter
the level of improvement that we observe. This is because the primary
bottleneck is the latency involved in locating the file on disk ---
when file sizes are small, the time needed to actually load the file from disk
(which is what the specifics of the file size distribution affect) is negligible.

Second, as predicted in our queueing model (\S\ref{sec:queueing-model}),
increasing the variability in the system causes redundancy to perform
better. We tried increasing variability in two different ways ---
increasing the proportion of access hitting disk by reducing the
cache-to-disk ratio (Figure~\ref{fig:ctd1t100}), and running on a
public cloud (EC2) instead of dedicated hardware
(Figure~\ref{fig:ec2}).  The increase in improvement is relatively
minor, although still noticeable, when we reduce the cache-to-disk
ratio. The benefit is most visible in the tail: the $99.9$th
percentile latency improvement at $10\%$ load goes up from $2.3\times$
in the base configuration to $2.8\times$ when we use the smaller
cache-to-disk ratio, and from $2.2\times$ to $2.5\times$ at $20\%$
load.

The improvement is rather more dramatic when going from Emulab to
EC2. Redundancy cuts the mean response time at $10$-$20\%$ load on EC2
in half, from $12$ ms to $6$ ms (compare to the $1.3-1.5\times$
reduction on Emulab). The tail improvement is even larger: on EC2, the
$99.9$th percentile latency at $10$-$20\%$ load drops $8\times$ when
we use redundancy, from around $160$ ms to $20$ ms. It is noteworthy
that the worst $0.1\%$ of \emph{outliers} with replication are quite close to
the $12$ ms \emph{mean} without replication!

Third, as also predicted in \S\ref{sec:queueing-model}, redundancy ceases to
help when the client-side overhead due to replication is a significant
fraction of the mean service time, as is the case when the file
sizes are very large (Figure~\ref{fig:filesize400KB}) or when the
cache is large enough that all the files fit in memory
(Figure~\ref{fig:ctd2t1}). We study this second scenario more
directly, using an in-memory distributed database, in the next
section.

\subsection{Application: memcached}
\label{sec:memcached}

We run a similar experiment to the one in the previous section, except
that we replace the filesystem store + Linux kernel cache + Apache web
server interface setup with the memcached in-memory
database. Figure~\ref{fig:memcached} shows the observed response times
in an Emulab deployment. The results show that replication seems to
worsen overall performance at all the load levels we tested ($10$-$90\%$).

\begin{figure*}
  \centering
  \subfigure{
  \includegraphics[width=0.3\textwidth]{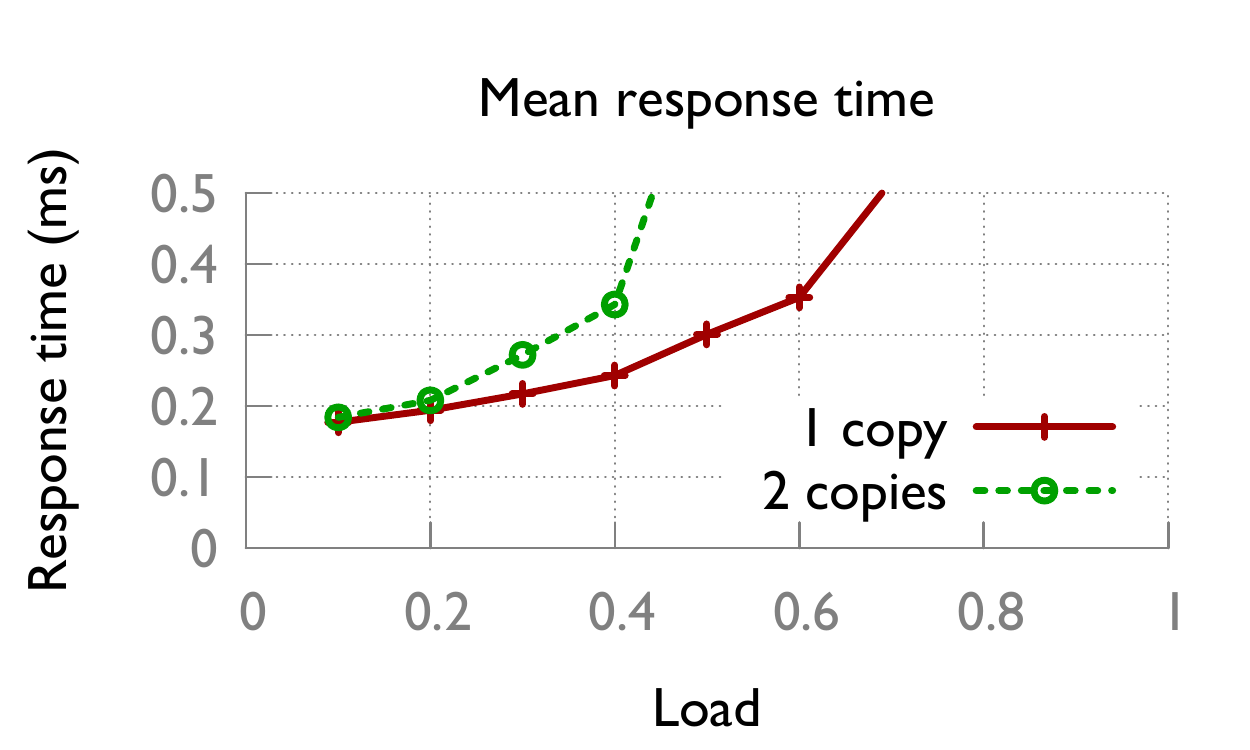}
  }
  \subfigure{
  \includegraphics[width=0.3\textwidth]{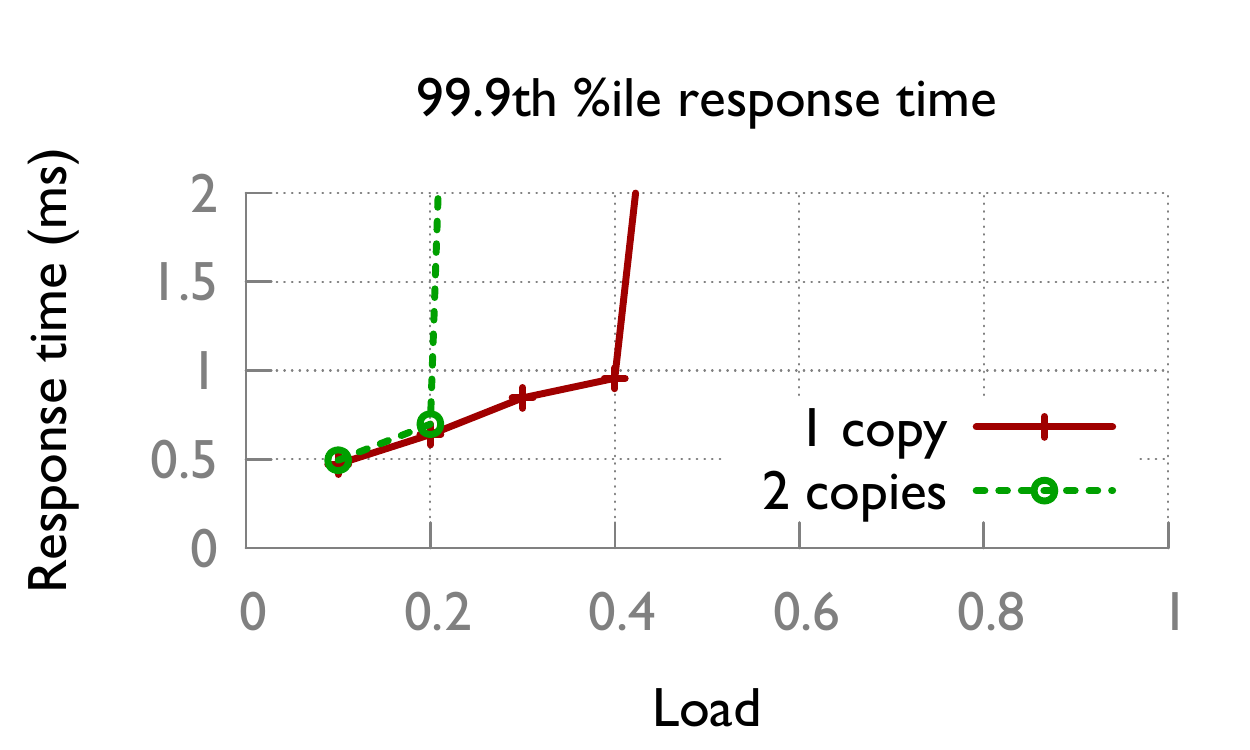}
  }
  \subfigure{
  \includegraphics[width=0.3\textwidth]{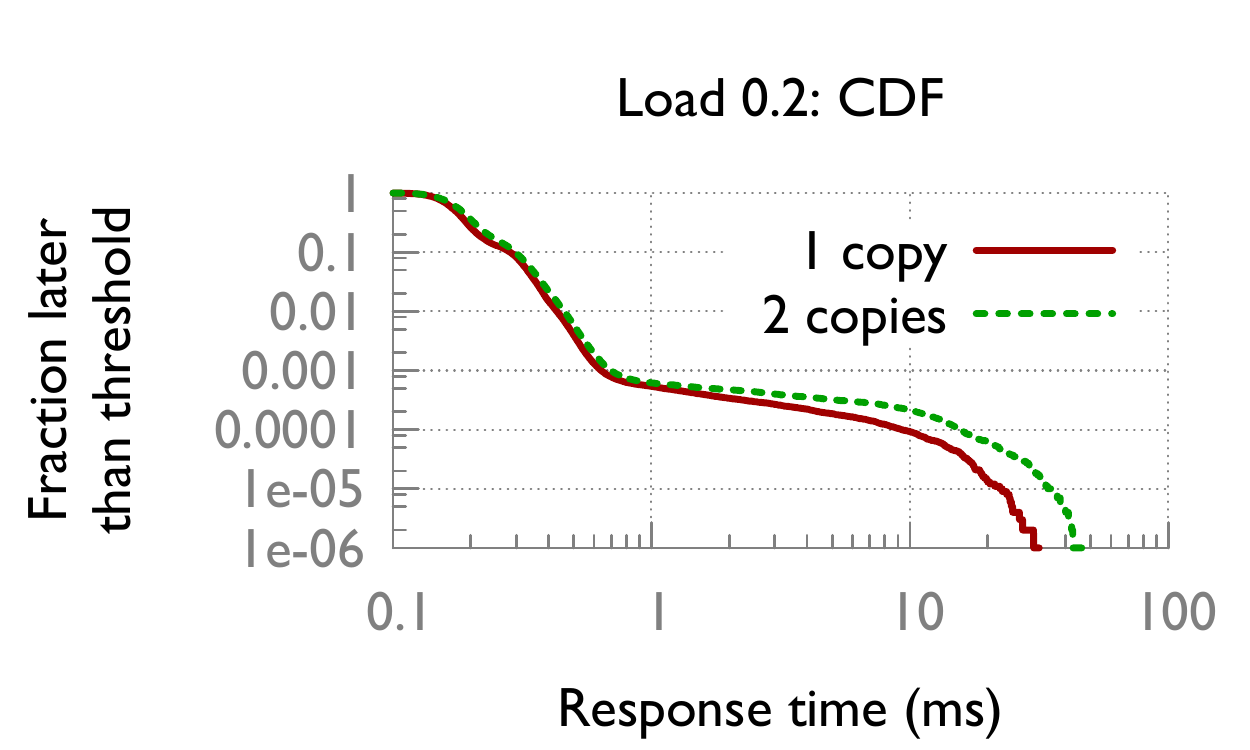}
  }
  \caption{memcached}
  \label{fig:memcached}
\end{figure*}
 
To understand why, we test two versions of our code at a low ($0.1\%$) load
level: the ``normal'' version, as well as a version with the
calls to memcached replaced with stubs, no-ops that return
immediately. The performance of this stub version is an estimate of
how much client-side latency is involved in processing a query.

\begin{figure}
\centering
\includegraphics[width=0.45\textwidth]{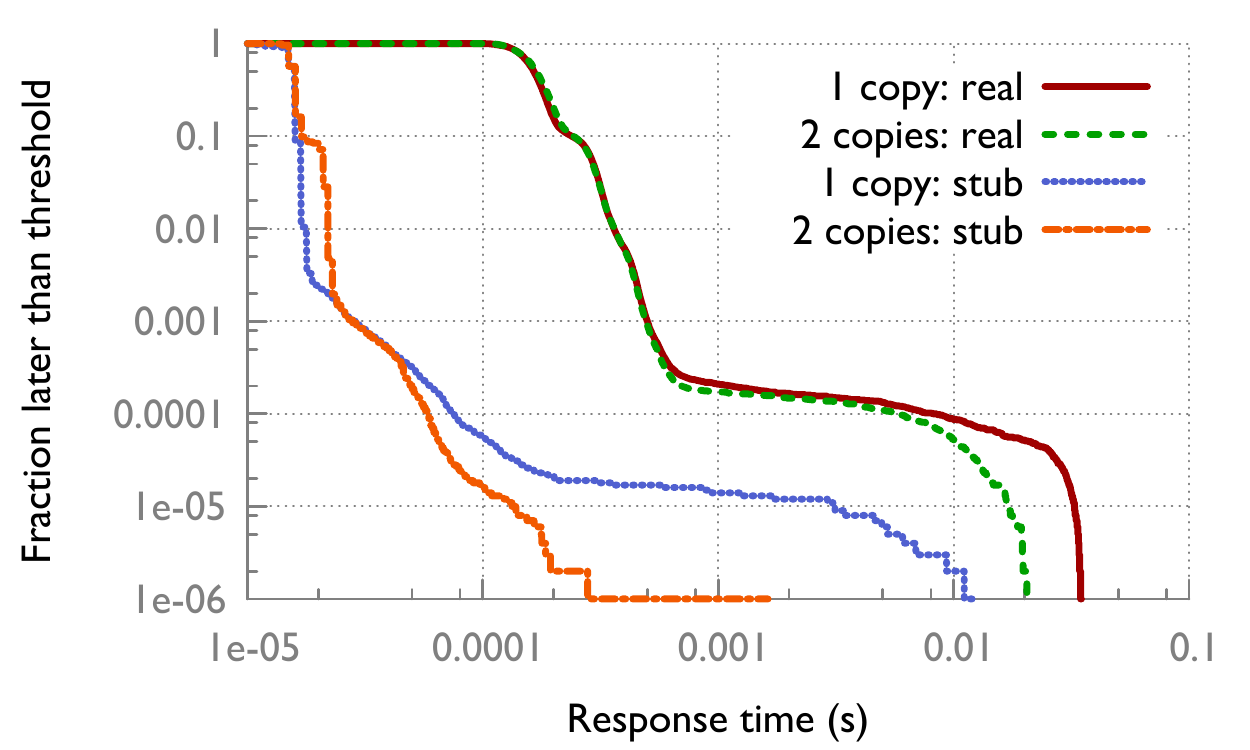}
\caption{memcached: stub and normal version response times at $0.1\%$
  load}
\label{fig:memcached-stub}
\end{figure}

Figure~\ref{fig:memcached-stub} shows that the client-side latency is
non-trivial. Replication increases the mean response time in the stub
version by $0.016$ ms, which is $9\%$ of the $0.18$ ms mean service
time. This is an underestimate of the true client-side overhead since
the stub version, which doesn't actually process queries, does not
measure the network and kernel overhead involved in sending and
receiving packets over the network.

The client-side latency overhead due to redundancy is thus at least
$9\%$ of the mean service time. Further, the service time distribution
is not very variable: although there are outliers, more than
$99.9\%$ of the mass of the entire distribution is within a factor of
$4$ of the mean.  Figure~\ref{fig:client-overhead} in
\S\ref{sec:queueing-model} shows that when the service time
distribution is completely deterministic, a client-side overhead
greater than $3\%$ of the mean service time is large enough to
completely negate the response time reduction due to redundancy.

In our system, redundancy does not seem to have that absolute a
negative effect -- in the ``normal'' version of the code, redundancy
still has a slightly positive effect overall at $0.1\%$ load
(Figure~\ref{fig:memcached-stub}). This suggests that the threshold
load is positive though small (it has to
be smaller than $10\%$: Figure~\ref{fig:memcached} shows that
replication always worsens performance beyond $10\%$ load).


\subsection{Application: replication in the network}

Replication has always added a non-zero amount of overhead in the systems we have considered so far (even if that overhead was mitigated by the response time reduction it achieved). We now consider a setting in which this overhead can be essentially eliminated: a network in which the switches are capable of strict prioritization.

Specifically, we consider a data center network. Many data center network architectures~\cite{Greenberg2009,dctcp} provide multiple equal-length paths between each source-destination pair, and assign flows to paths based on a hash of the flow header~\cite{rfc2992}. However, simple static flow assignment interacts poorly with the highly skewed flow-size mix typical of data centers: the majority of the traffic volume in a data center comes from a small number of large elephant flows~\cite{Alizadeh2012,dctcp}, and hash-based flow assignment can lead to hotspots because of the possibility of assigning multiple elephant flows to the same link, which can result in significant congestion on that link. Recent work has proposed mitigating this problem by dynamically reassigning flows in response to hotspots, in either a centralized~\cite{Al-Fares2010} or distributed~\cite{Wu2012} fashion.

We consider a simple alternative here: redundancy. Every switch replicates the first few packets of each flow along an alternate route, reducing the probability of collision with an elephant flow. Replicated packets are assigned a lower (strict) priority than the original packets, meaning they can never delay the original, unreplicated traffic in the network. Note that we could, in principle, replicate \emph{every} packet --- the performance when we do this can never be worse than without replication --- but we do not since unnecessary replication can reduce the gains we achieve by increasing the amount of queueing \emph{within} the replicated traffic. We replicate only the first few packets instead, with the aim of reducing the latency for short flows (the completion times of large flows depend on their aggregate throughput rather than individual per-packet latencies, so replication would be of little use).

We evaluate this scheme using an ns-3 simulation of a common $54$-server three-layered fat-tree topology, with a full bisection-bandwidth fabric consisting of $45$ $6$-port switches organized in $6$ pods. We use a queue buffer size of $225$ KB and vary the link capacity and delay. Flow arrivals are Poisson, and flow sizes are distributed according to a standard data center workload~\cite{Benson2010}, with flow sizes varying from $1$ KB to $3$ MB and with more than $80\%$ of the flows being less than $10$ KB.

\begin{figure*}[t]
\centering
\subfigure{
\includegraphics[width=0.3\textwidth]{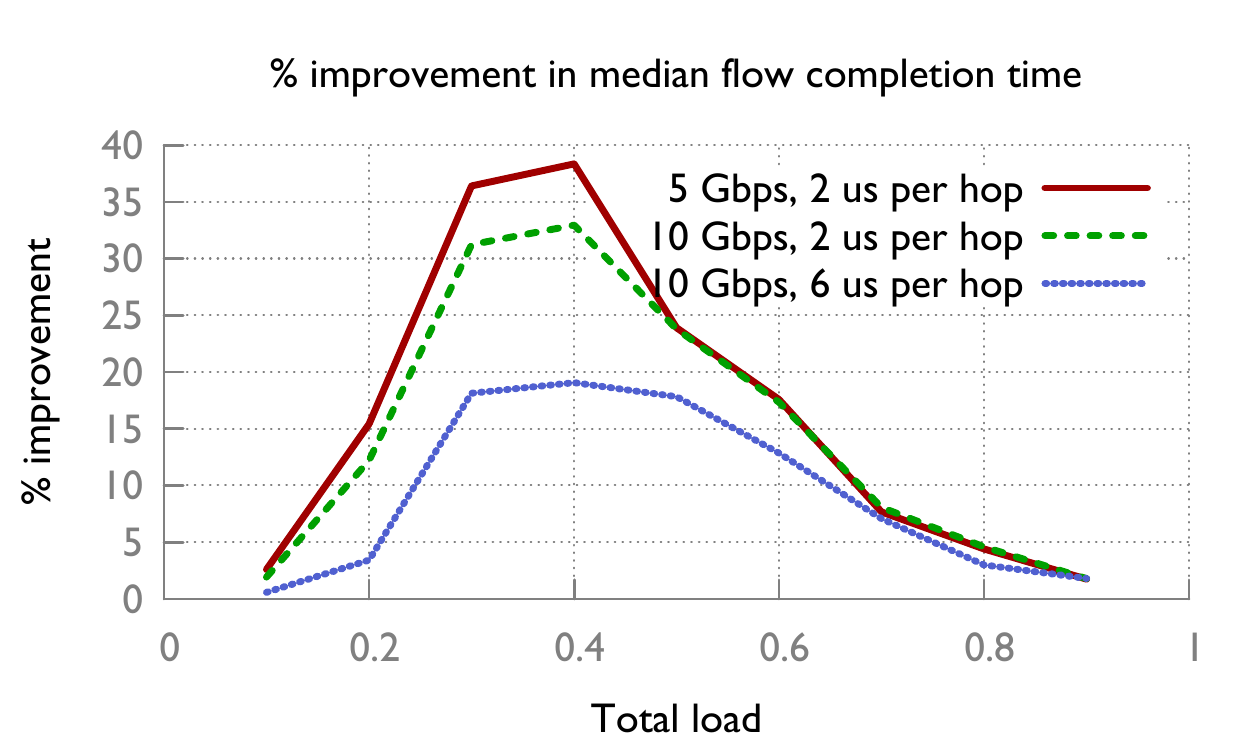}
\label{fig:network-median-trend}
}
\subfigure{
\includegraphics[width=0.3\textwidth]{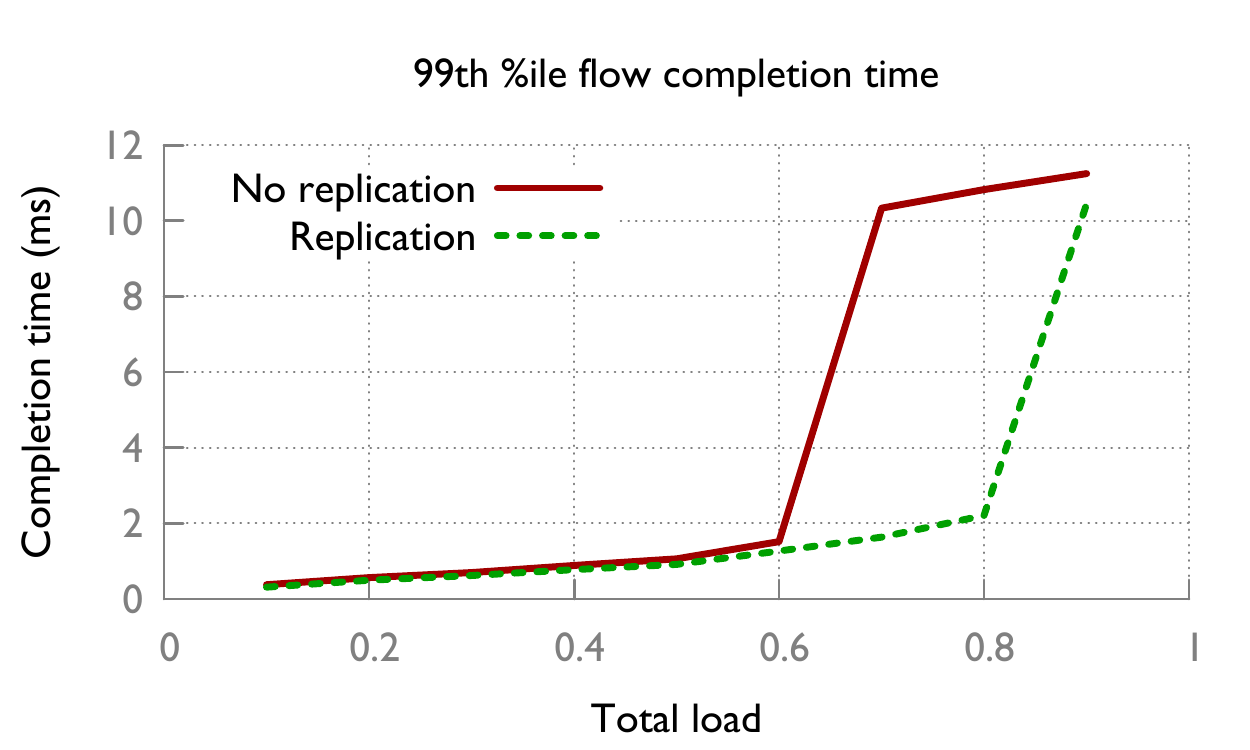}
\label{fig:network-p99}
}
\subfigure{
\includegraphics[width=0.3\textwidth]{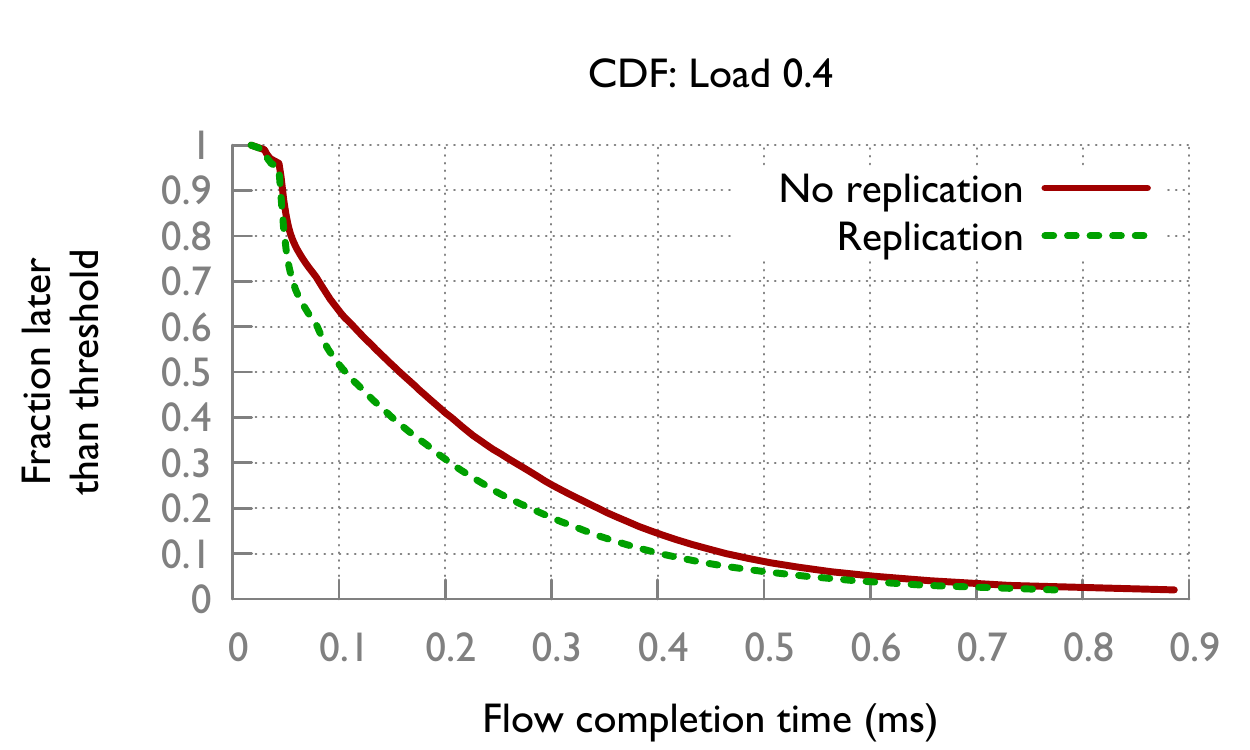}
\label{fig:network-cdf}
}
\caption{Median and tail completion times for flows smaller than $10$ KB}
\label{fig:network}
\end{figure*}

Figure~\ref{fig:network} shows the completion times of flows smaller than $10$ KB when we replicate the first $8$ packets in every flow.

Figure~\ref{fig:network-median-trend} shows the reduction in the median flow completion time as a function of load for three different delay-bandwidth combinations. Note that in all three cases, the improvement is small at low loads, rises until load $\approx 40\%$, and then starts to fall. This is because at very low loads, the congestion on the default path is small enough that replication does not add a significant benefit, while at very high loads, every path in the network is likely to be congested, meaning that replication again yields limited gain. We therefore obtain the largest improvement at intermediate loads.

Note also that the performance improvement we achieve falls as the delay-bandwidth product increases. This is because our gains come from the reduction in queuing delay when the replicated packets follow an alternate, less congested, route. At higher delay-bandwidth products, queueing delay makes up a smaller proportion of the total flow completion time, meaning that the total latency savings achieved is correspondingly smaller. At $40\%$ network load, we obtain a $38\%$ improvement in median flow completion time ($0.29$ ms vs.\ $0.18$ ms) when we use $5$ Gbps links with $2$ us per-hop delay. The improvement falls to $33\%$ ($0.15$ ms vs.\ $0.10$ ms) with $10$ Gbps links with $2$ us per-hop delay, and further to $19\%$ ($0.21$ ms vs.\ $0.17$ ms) with $10$ Gbps links with $6$ us per-hop delay.


Next, Figure~\ref{fig:network-p99} shows the $99$th percentile flow completion times for one particular delay-bandwidth combination. In general, we see a $10$-$20\%$ reduction in the flow completion times, but at $70$-$80\%$ load, the improvement spikes to $80$-$90\%$. The reason turns out to be timeout avoidance: at these load levels, the $99$th percentile unreplicated flow faces a timeout, and thus has a completion time greater than the TCP minRTO, $10$ ms. With redundancy, the number of flows that face timeouts reduces significantly, causing the $99$th percentile flow completion time to be much smaller than $10$ ms.

At loads higher than $80\%$, however, the number of flows facing timeouts is high even with redundancy, resulting in a narrowing of the performance gap.

Finally, Figure~\ref{fig:network-cdf} shows a CDF of the flow completion times at one particular load level. Note that the improvement in the mean and median is much larger than that in the tail. We believe this is because the high latencies in the tail occur at those instants of high congestion when most of the links along the flow's default path are congested. Therefore, the replicated packets, which likely traverse some of the same links, do not fare significantly better.

Replication has a negligible impact on the elephant flows: it improved the mean completion time for flows larger than $1$ MB by a statistically-insignificant $0.12\%$.






\section{Individual view}
\label{sec:individual-view}

The model and experiments of the previous section indicated that in a
range of scenarios, latency is best optimized in a fixed set of system
resources through replication. However, \new{settings such as the
  wide-area Internet} are better modeled as having \emph{elastic}
resources: individual participants can \new{selfishly} choose whether
to replicate an operation, but this incurs an additional cost
\new{(such as bandwidth usage or battery consumption). In this
  section, we present two examples of wide-area Internet applications
  in which replication achieves a substantial improvement in
  latency. We argue that the latency reduction in both these
  applications outweighs the cost of the added overhead by comparing
  against a benchmark due to Vulimiri et al.~\cite{vulimiri12latency, ArxivAppendix}, who
  computed a cost-effectiveness threshold by comparing the cost of the
  extra processing that would be induced at the clients and the
  servers against the economic value of the latency improvement that
  would be achieved.}


%
\subsection{Application: Connection establishment}
\label{sec:tcp-handshake}

We start with a simple example, demonstrating why replication should be cost-effective even when the available choices are limited: we \new{use a back-of-the-envelope calculation} to consider what happens when multiple copies of TCP-handshake packets are sent on the same path. It is obvious that this should help if all packet losses on the path are independent. In this case, sending two back-to-back copies of a packet would reduce the probability of it being lost from $p$ to $p^2$. In practice, of course, back-to-back packet transmissions are likely to observe a correlated loss pattern. But Chan et al.~\cite{Chan2010} measured a significant reduction in loss probability despite this correlation. Sending back-to-back packet pairs between PlanetLab hosts, they found that the average probability of individual packet loss was $\approx 0.0048$, and the probability of \emph{both} packets in a back-to-back pair being dropped was only $\approx 0.0007$ -- much larger than the $\sim 10^{-6}$ that would be expected if the losses were independent, but still $7\times$ lower than the individual packet loss rate.\footnote{It might be possible to do even better by spacing the transmissions of the two packets in the pair a few milliseconds apart to reduce the correlation.}

As a concrete example, we quantify the improvement that this loss rate reduction would effect on the time required to complete a TCP handshake. The three packets in the handshake are ideal candidates for replication: they make up an insignificant fraction of the total traffic in the network, and there is a high penalty associated with their being lost (Linux and Windows use a $3$ second initial timeout for SYN packets; OS X uses $1$ second~\cite{Chu2009}). We use the loss probability statistics discussed above to estimate the expected latency savings on each handshake.

We consider an idealized network model. Whenever a packet is sent on the network, we assume it is delivered successfully after $(RTT/2)$ seconds with probability $1 - p$, and lost with probability $p$. Packet deliveries are assumed to be independent of each other. $p$ is $0.0048$ when sending one copy of each packet, and $0.0007$ when sending two copies of each packet. We also assume TCP behavior as in the Linux kernel: an initial timeout of $3$ seconds for SYN and SYN-ACK packets and of $3\times RTT$ for ACK packets, and exponential backoff on packet loss~\cite{Chu2009}.

With this model, it can be shown that duplicating all three packets in the handshake would reduce its expected completion time by approximately $(3 + 3 + 3 \times RTT) \times (4.8 - 0.7)$ ms, which is at least $25$ ms.  \new{The benefit increases with $RTT$, and is even higher in the tail: duplication would improve the $99.9$th percentile handshake completion time by at least $880$ ms.}

\new{Is this improvement worth the cost of added traffic?  Qualitatively, even $25$ ms is significant relative to the size of the handshake packets.  Quantitatively, a cost-benefit analysis is difficult since it depends on estimating and relating the direct and indirect costs of added traffic and the value to humans of lower latency.  While an accurate comparison is likely quite difficult, the study referenced at the beginning of this section~\cite{ArxivAppendix, vulimiri12latency} estimated these values using pricing of cloud services, which encompass a broad range of costs beyond only bandwidth, and concluded that in a broad class of cases, reducing latency is useful as long as it improves latency by $16$ ms for every KB of extra traffic.  In comparison, the latency savings we obtain in TCP connection establishment is more than an order of magnitude larger than this threshold in the mean, and more than two orders of magnitude larger in the tail.  Specifically, if we assume each packet is $50$ bytes long then a $25$-$880$ ms improvement implies a savings of around $170$-$6000$ ms/KB.  We caution, however, that the analysis of~\cite{ArxivAppendix, vulimiri12latency} was necessarily imprecise; a more rigorous study would be an interesting avenue of future work.}



\subsection{Application: DNS}
\label{sec:dns}


An ideal candidate for replication is a service that involves small operations and which is replicated at multiple locations, thus providing diversity across network paths and servers, so that replicated operations are quite independent.  We believe opportunities to replicate queries to such services may arise both in the wide area and the data center.  Here, we explore the case of replicating DNS queries.

We began with a list of 10 DNS servers\footnote{The default local DNS server, plus public servers from Level3, Google, Comodo, OpenDNS, DNS Advantage, Norton DNS, ScrubIT, OpenNIC, and SmartViper.} and Alexa.com's list of the top 1 million website names. At each of $15$ PlanetLab nodes across the continental US, we ran a two-stage experiment: (1) Rank all 10 DNS servers in terms of mean response time, by repeatedly querying a random name at a random server. Note that this ranking is specific to each PlanetLab server.  (2) Repeatedly pick a random name and perform a random one of 20 possible trials --- either querying one of the ten individual DNS servers, or querying anywhere from 1 to 10 of the best servers in parallel (e.g. if sending 3 copies of the query, we send them to the top 3 DNS servers in the ranked list).  In each of the two stages, we performed one trial every 5 seconds. We ran each stage for about a week at each of the $15$ nodes.  Any query which took more than $2$ seconds was treated as lost, and counted as $2$ sec when calculating mean response time.

\begin{figure}[t]
\centering
\includegraphics[width=0.8\columnwidth]{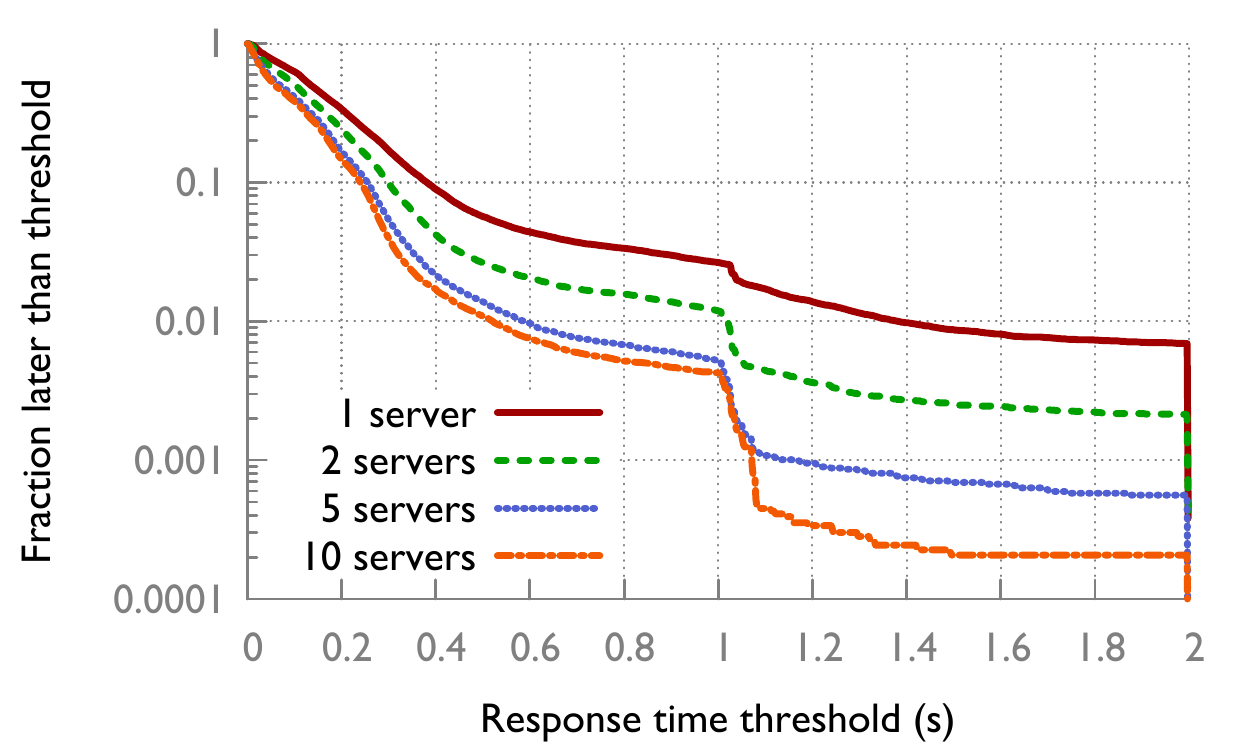}
\caption{DNS response time distribution.}
\label{fig:dns:ccdf}
\end{figure}

\begin{figure}[t]
\centering
\includegraphics[width=0.8\columnwidth]{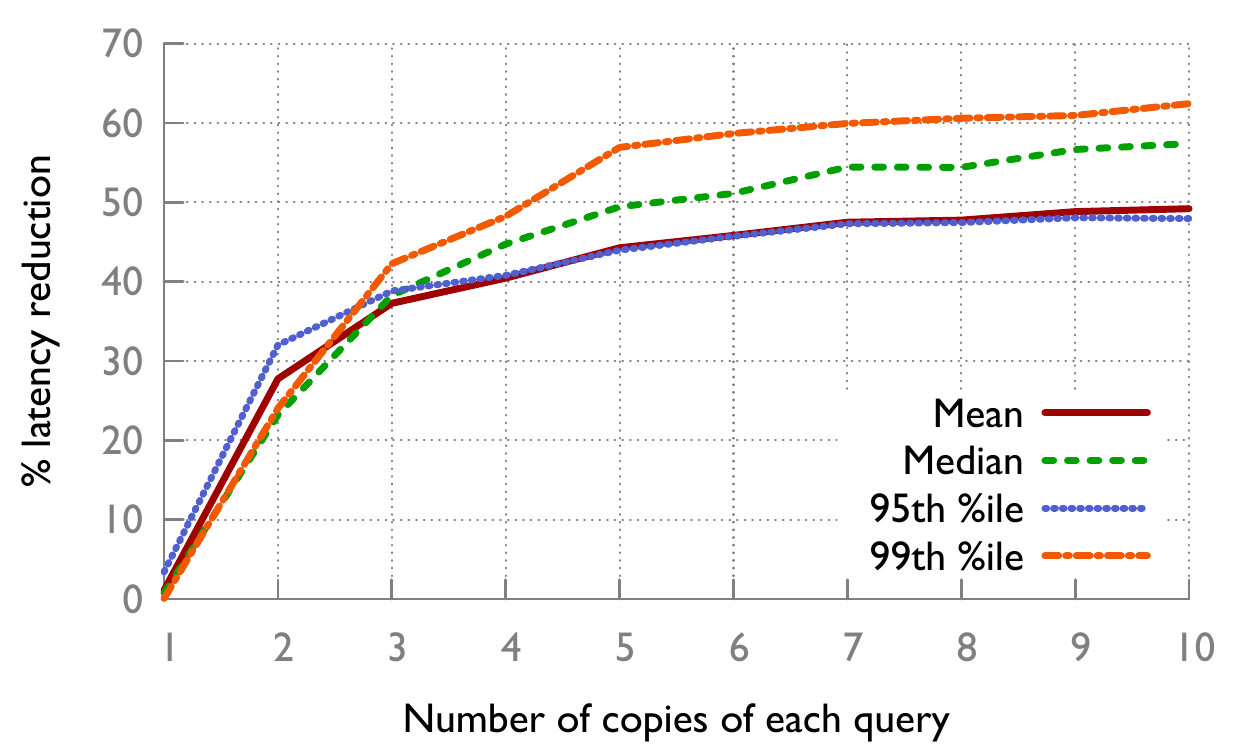}
\caption{Reduction in DNS response time, averaged across $15$ PlanetLab servers.}
\label{fig:dns:fixed}
\end{figure}
%

Figure~\ref{fig:dns:ccdf} shows the distribution of query response times across all the PlanetLab nodes.  The improvement is substantial, especially in the tail: Querying $10$ DNS servers, the fraction of queries later than $500$ ms is reduced by $6.5\times$, and the fraction later than $1.5$ sec is reduced by $50\times$.  Averaging over all PlanetLab nodes, Figure~\ref{fig:dns:fixed} shows the average percent reduction in response times compared to the best fixed DNS server identified in stage 1.  We obtain a substantial reduction with just $2$ DNS servers in all metrics, improving to $50$-$62$\% reduction with $10$ servers.  Finally,
we compared performance
to the best single server \emph{in retrospect}, i.e., the server with minimum mean response time for the queries to individual servers in Stage 2 of the experiment, since the best server may change over time.  Even compared with this stringent baseline, we found a result similar to Fig.~\ref{fig:dns:fixed}, with a reduction of $44$-$57$\% in the metrics when querying $10$ DNS servers.

\begin{figure}
\centering
\includegraphics[width=0.95\columnwidth]{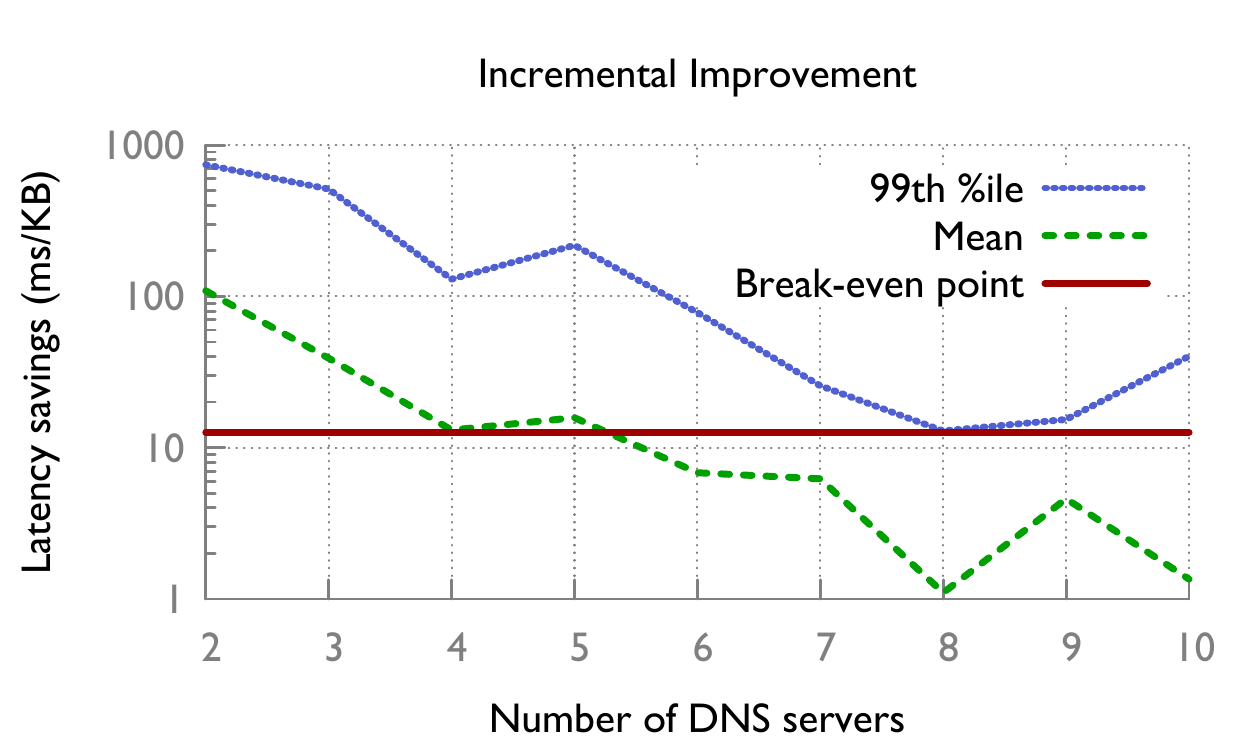}
\caption{Incremental latency improvement from each extra server contacted}
\label{fig:dns_marginal_savings}
\end{figure}

How many servers should one use? Figure~\ref{fig:dns_marginal_savings} compares the marginal increase in latency savings from each extra server against the $16$ ms/KB benchmark~\cite{ArxivAppendix, vulimiri12latency} discussed earlier in this section. The results show that what we should do depends on the metric we care about. If we are only concerned with mean performance, it does not make economic sense to contact any more than $5$ DNS servers for each query, but if we care about the $99$th percentile, then it is always useful to contact $10$ or more DNS servers for every query. Note also that the \emph{absolute} (as opposed to the marginal) latency savings is still worthwhile, even in the mean, if we contact $10$ DNS servers for every query. The absolute mean latency savings from sending $10$ copies of every query is $0.1$ sec $/$ $4500$ extra bytes $\approx$ $23$ ms/KB, which is more than twice the break-even latency savings.  And if the client costs are based on DSL rather than cell service, the above schemes are all more than $100\times$ more cost-effective.


Querying multiple servers also increases caching, a side-benefit which would be interesting to quantify.

Prefetching --- that is, preemptively initiating DNS lookups for all links on the current web page --- makes a similar tradeoff of increasing load to reduce latency, and its use is widespread in web browsers. Note, however, that redundancy is complementary to prefetching, since some names in a page will not have been present on the previous page (or there may not be a previous page).

\section{Related work}
\label{sec:related}

Replication is used pervasively to improve reliability, and in many
systems to reduce latency. Distributed job execution frameworks, for
example, have used task replication to improve response time, both
preemptively~\cite{Foster72, ananthanarayanan12dolly} and to mitigate
the impact of stragglers~\cite{Zaharia2008}.

Within networking, replication has been explored to reduce latency in
several specialized settings, including replicating DHT queries to
multiple servers~\cite{li2005bandwidth} and replicating transmissions
(via erasure coding) to reduce delivery time and loss probability in
delay-tolerant
networks~\cite{jain2005using,soljanin2010reducing}. Replication has
also been suggested as a way of providing QoS prioritization and
improving latency and loss performance in networks capable of
redundancy elimination~\cite{han2012rpt}.

Dean and Barroso~\cite{Dean2013} discussed Google's use of redundancy
in various systems, including a storage service similar to the one we
evaluated in \S\ref{sec:apache}, but they studied specific systems
with capabilities that are not necessarily available in general (such
as the ability to cancel outstanding partially-completed requests),
and did not consider the effect the total system utilization could
have on the efficacy of redundancy. In contrast, we thoroughly
evaluate the effect of redundancy at a range of loads both in various
configurations of a deployed system (\S\ref{sec:apache},
\S\ref{sec:memcached}), and in a large space of synthetic scenarios in an
abstract system model (\S\ref{sec:queueing-model}).

Andersen et al.~\cite{andersen05monet}'s
MONET system proxies web traffic through an overlay network formed out
of multi-homed proxy servers. While the primary focus of
\cite{andersen05monet} is on adapting quickly to changes in path
performance, they replicate two specific subsets of their traffic:
connection establishment requests to multiple servers are sent in
parallel (while the first one to respond is used), and DNS queries are
replicated to the local DNS server on each of the multi-homed proxy
server's interfaces. We show that replication can be useful in both
these contexts even in the absence of path diversity: a significant
performance benefit can be obtained by sending multiple copies of TCP
SYNs to the \emph{same} server on the \emph{same} path, and by
replicating DNS queries to multiple public servers over the \emph{same} access link.

In a recent workshop paper~\cite{vulimiri12latency} \new{we} advocated using redundancy to reduce latency, but
it was preliminary work that did not characterize when redundancy is helpful, and did not study the systems view of optimizing a fixed set of resources.

Most importantly, unlike all of the above work, our goal is to
demonstrate the power of redundancy as a general technique. We do this
by providing a characterization of when it is (and isn't) useful, and
by quantifying the performance improvement it offers in several
use cases where it is applicable.


\section{Conclusion}

We studied an abstract characterization of the tradeoff
between the latency reduction achieved by redundancy and the cost of
the overhead it induces to demonstrate that redundancy should have a net positive
impact in a large class of systems. We then confirmed empirically that
redundancy offers a significant benefit in a number of practical
applications, both in the wide area and in the data center. We believe
our results demonstrate that redundancy is a powerful technique that
should be used much more commonly in networked systems than it
currently is.  Our results also will guide the \emph{judicious} application
of redundancy within only those cases where it is a win in terms of
performance or cost-effectiveness.





\bibliographystyle{abbrv}

\footnotesize
\bibliography{paper}
%

\end{document}